\documentclass[12pt]{article}
\usepackage{amsmath}
\usepackage{graphicx}
\usepackage{enumerate}
\usepackage{url} 

\newcommand{\blind}{1}

\usepackage{amsfonts}%
\usepackage{amssymb}%
\usepackage{amsthm}
\usepackage{bm}
\usepackage{relsize}
\usepackage{graphicx}
\usepackage{caption}
\usepackage{epstopdf}
\usepackage{tikz}
\usetikzlibrary{trees,arrows}
\usetikzlibrary{decorations}
\usetikzlibrary[decorations]
\usepgflibrary{decorations.pathmorphing}
\usepgflibrary[decorations.pathmorphing]
\usetikzlibrary{decorations.pathmorphing}
\usetikzlibrary[decorations.pathmorphing]
\usepackage{booktabs}
\usepackage[authoryear]{natbib}
\usepackage{subcaption}
\usepackage{pseudocode}
\usepackage{verbatim} 

\bibpunct{(}{)}{,}{}{}{;} 

\newcommand\independent{\protect\mathpalette{\protect\independenT}{\perp}}
\def\independenT#1#2{\mathrel{\rlap{$#1#2$}\mkern2mu{#1#2}}}

\DeclareMathOperator*{\E}{\mathbb{E}}
\DeclareMathOperator*{\Et}{\mathbb{E}_t}

\DeclareMathOperator{\circlearrow}{\hbox{$\circ$}\kern-1.5pt\hbox{$\rightarrow$}}
\DeclareMathOperator{\circlecircle}{\hbox{$\circ$}\kern-1.2pt\hbox{$--$}\kern-1.5pt\hbox{$\circ$}}

\DeclareMathOperator*{\Odds}{\textrm{OR}}
\DeclareMathOperator*{\Oddst}{\textrm{OR}_t}
\DeclareMathOperator{\logit}{logit}

\newtheorem{Lma}{Lemma}
\newtheorem{Thm}{Theorem}
\newtheorem{Cor}{Corollary}

\begin{document}

\def\spacingset#1{\renewcommand{\baselinestretch}%
{#1}\small\normalsize} \spacingset{1}


\if1\blind
{
  \title{\bf Semiparametric Inference for Non-monotone Missing-Not-at-Random Data: the No Self-Censoring Model}
  \author{Daniel Malinsky,\thanks{
    The authors gratefully acknowledge funding from the National Institutes of Health, grant R01 AI127271-01 A1.}\\
    Department of Biostatistics, Columbia University\\
    Ilya Shpitser,\\
    Department of Computer Science, Johns Hopkins University\\
    and \\
    Eric J. Tchetgen Tchetgen \\
    Department of Statistics, The Wharton School of the University of Pennsylvania}
  \maketitle
} \fi

\if0\blind
{
  \bigskip
  \bigskip
  \bigskip
  \begin{center}
    {\LARGE\bf Semiparametric Inference for Non-monotone Missing-Not-at-Random Data: the No Self-Censoring Model}
\end{center}
  \medskip
} \fi

\bigskip
\begin{abstract}
We study the identification and estimation of statistical functionals of multivariate data missing non-monotonically and not-at-random, taking a semiparametric approach. Specifically, we assume that the missingness mechanism satisfies what has been previously called ``no self-censoring'' or ``itemwise conditionally independent nonresponse," which roughly corresponds to the assumption that no partially-observed variable directly determines its own missingness status. We show that this assumption, combined with an odds ratio parameterization of the joint density, enables identification of functionals of interest, and we establish the semiparametric efficiency bound for the nonparametric model satisfying this assumption. We propose a practical augmented inverse probability weighted estimator, and in the setting with a (possibly high-dimensional) always-observed subset of covariates, our proposed estimator enjoys a certain double-robustness property. We explore the performance of our estimator with simulation experiments and on a previously-studied data set of HIV-positive mothers in Botswana.
\end{abstract}

\noindent%
{\it Keywords:}  Missing data; MNAR; Identification; Double-robustness
\vfill

\newpage
\spacingset{1.5} 
\section{Introduction}

Missing data are a pervasive feature of observations in almost every area of scientific study. Techniques for accounting for missing data in settings where the missingness depends only on observed data (``missingness-at-random" or MAR) are well-developed \citep{robins1994estimation, tsiatis2006semiparametric, little2014statistical}. The situation is substantially more difficult, however, in multivariate settings when the probability of missingness may depend on unobserved parts of the data (``missingness-not-at-random'' or MNAR) and when the patterns of missingness are non-monotone
\citep[see][]{robins1997nonignorable, rotnitzky1997analysis, scharfstein1999adjusting}. Non-monotone missingness may occur, for example, when there are complex patterns of drop-out/re-entry in a longitudinal study or when, 
as is often the case in practice, exposure, covariate, and outcome variables may each be subject to missing values with no specific pattern across the sample.
\cite{robinsgill1997ignorable} and \cite{vansteelandt2007estimation} argue that missingness-not-at-random should be expected in non-monotone longitudinal settings if, for example, a research subject's decision to re-enter depends in part on the evolution of attributes that would have been recorded in missed visits. MNAR is also common in settings where social stigma makes non-response to some research questions (e.g., about HIV status, sexual activity, or drug use) dependent on other imperfectly observed or censored questions \citep{marra2017simultaneous}.

Recent work on nonparametric or semiparametric inference in non-monotone MNAR settings has proceeded by positing some set of restrictions on the missingness mechanism sufficient for identifying a functional or parameter of interest \citep{rotnitzky1998semiparametric, robins2000sensitivity, vansteelandt2007estimation, zhou2010block, li2013weighting, shpitser2016consistent, sadinle2017itemwise, tchetgen2018discrete}.  We adopt the identifying assumption introduced in \cite{shpitser2016consistent} and \cite{sadinle2017itemwise} -- called ``no self-censoring'' in the former and ``itemwise conditionally independent nonresponse'' in the latter --  which allows for both missingness-not-at-random and non-monotonicity. Specifically, we assume only that each measured but sometimes missing variable is conditionally independent of its missingness indicator given all other variables (which may also be missing) and all other missingness indicators. Mechanistically interpreted, this means that no variable is a direct cause of its own missingness status. Our parameter of interest is 
any measurable function of the full data distribution (for example, a marginal mean, correlation, or regression parameter), which is identified from the observed data under this assumption.

\cite{shpitser2016consistent} introduces a pseudolikelihood-based inverse probability weighted (IPW) estimator for the ``no self-censoring'' model. His approach is consistent but (as is common for IPW estimators) not efficient and relies on a specific parameterization of the missingness mechanism. \cite{sadinle2017itemwise} introduce a modeling strategy that requires specifying joint distributions or multivariate kernel density estimation, which can be prohibitively challenging in practice and lacks desirable inferential guarantees such as $\sqrt{n}$-consistency and asymptotic normality for estimating a particular parameter of interest. In contrast to these approaches, we present a semiparametric analysis yielding influence function (IF) based estimators which have a number of desirable properties \citep{newey1990semiparametric, bickel1993efficient, van2000asymptotic, tsiatis2006semiparametric}. Furthermore, our approach exploits an odds ratio parameterization of joint densities due to \cite{chen2007semiparametric, chen2010compatibility}, thereby enabling convenient and congenial specification of the various components of the likelihood. Finally, the estimator we propose benefits from a certain appealing double-robustness property, which can mitigate the threat of model misspecification in the setting where a possibly high-dimensional set of always-observed covariates is also available.

We begin by introducing our central assumption, which we show implies nonparametric identification of both the probability of each missingness pattern and our parameter of interest. Next we use semiparametric theory to derive an observed data influence function of a pathwise differentiable functional on a nonparametric full data model.  Because the no self-censoring model is nonparametric saturated (as defined by \cite{robins1997nonignorable}), this influence function is unique and efficient.
We then consider the case where there is available an additional vector of always-observed covariates and demonstrate that our proposed estimator is doubly-robust. We explore the performance of our estimator in simulated data, and conclude with an application to a cohort study of HIV-positive mothers in Botswana.

\section{The model}

Suppose the underlying data-generating process yields i.i.d.\ samples of $(R,L)$ with full data vector $L = (L_1,...,L_K)'$ and missingness indicators $R = (R_1,...,R_K)'$. $R$ takes values in $\{0,1\}^K$ where $1$ corresponds to ``observed'' and $0$ corresponds to ``missing.'' That is, $R_i = 1$ if $L_i$ is observed and $R_i=0$ otherwise. $L$ may be continuous or discrete. In slight abuse of notation, we use equations such as $R=r$ and $R=1$ as shorthand for $(R_1,...,R_K) = (r_1,...,r_K)$ and $(R_1,...,R_K) = (1,...,1)^K$ (an identity vector of length $K$). Also for any vector $A$ we use $A_{-i} = (A_1,...,A_{i-1},A_{i+1},...,A_K)'$ to denote the vector $A$ with $i$th entry removed. Let $L_{(r)}$ be the subvector of the elements of $L$ that are observed when $R=r$. The observed data is comprised of i.i.d.\ realizations of the vector $(R, L_{(R)})$. 
We use $p(\cdot)$ to denote a distribution or density function.


It is well known that the full data distribution $p(L)$ is not identified from observed data distribution $p(R, L_{(R)})$ without a restriction on the missingness mechanism. We assume the following condition holds:\\ [2ex]
\textbf{Assumption 1} (no self-censoring)
\begin{equation} \label{eq:NSC}
R_i \independent L_i \mid R_{-i},L_{-i}
\end{equation}
for all $i=1,...,K$. 
We also make the following standard assumption:\\ [2ex]
\textbf{Assumption 2} (positivity)
\begin{equation} \label{eq:positivity}
p(R=1|L) > \sigma > 0
\end{equation}
w.p.1 for some constant $\sigma$. 

\cite{shpitser2016consistent} and \cite{sadinle2017itemwise} discuss the no self-censoring assumption (\ref{eq:NSC}) extensively. We only note a few features here which make the model interesting. First, the assumption does not place any restriction on the observed data distribution, i.e., the model is nonparametric saturated \cite[see][]{shpitser2016consistent, sadinle2017itemwise, sadinle2017sequential}. Second, if any of the $K$ independence assumptions in (\ref{eq:NSC}) are false, the joint distribution is no longer nonparametrically identified \citep{mohan2013graphical}. Thus, the model defined by (\ref{eq:NSC}) is an appealing starting point for analyzing MNAR data, either as a substantive model or in the course of sensitivity analysis.

Several authors have used graphical models to represent missingness models and establish identification results \citep{mohan2013graphical,shpitser2015missing,tian2015missing,mohan2018graphical,bhattacharya2019identification}. So far, these have focused mostly on missingness models that can be represented by directed acyclic graphs.
The independence model defined by (\ref{eq:NSC}) can be summarized graphically by a chain graph  \citep{lauritzen1996graphical}. A chain graph $\mathcal{G} = (V,E)$ with vertices $V$ and edges $E$ is a mixed graph that may contain both directed and undirected edges. We associate the vertices $V$ in a graph $\mathcal{G}$ with the random variables $(R_1,...,R_K, L_1,...,L_K)$. Specifically, the chain graph representation of the no self-censoring model contains undirected edges between all pairs of $R$ vertices, directed edges from each $L_i$ to each $R_j$ such that $j \neq i$, and directed edges among all $L$ vertices. None of the undirected edges among $R$ vertices may be directed without changing the model, however any of the directed edges among $L$ vertices may be reversed without changing the model. An example is shown for $K=3$ variables in Figure \ref{fig:chain}.

The absence of any edge between each $L_i$ and the corresponding $R_i$ corresponds to the independence assumption (\ref{eq:NSC}).\footnote{``Self-censoring'' edges $L_i \rightarrow R_i$ are also called ``self-masking'' in \cite{mohan2018estimation} and \cite{tu2019causal}.} This follows from the pairwise Markov property for chain graphs. A vertex $V_i$ is called a descendent of $V_j$ if there is a directed path from $V_j$ to $V_i$ in $\mathcal{G}$. Let $\text{nd}_i(\mathcal{G})$ denote the set of non-descendents of vertex $V_i$ in $\mathcal{G}$. The pairwise Markov property states that for non-adjacent vertices $V_i, V_j$ such that $V_j \in \text{nd}_i(\mathcal{G})$, $V_i \independent V_j \mid \text{nd}_i(\mathcal{G}) \setminus V_j$. Taking any $R_i,L_i$ in Figure \ref{fig:chain} to be $V_i,V_j$ we see this corresponds exactly to the no self-censoring assumption (\ref{eq:NSC}). Since there are no other absent edges (non-adjacencies) in $\mathcal{G}$, the model imposes no additional independence restrictions. The MNAR model corresponding to dropping undirected edges between missingness indicators, discussed in e.g.\ \cite{mohan2013graphical} and \cite{mohan2018graphical}, is a submodel of ours. \cite{mohan2018estimation} use graphical models to study identifiability when (\ref{eq:NSC}) is violated but additional parametric assumptions hold.
Though we do not emphasize graphical concepts in our results below, the graphical perspective can provide signficant insight into nonparametric identifiability. In particular, a version of our first identification result in the next section can be derived by inspecting the Markov factorization for a chain graph, which we describe in the appendix.

\cite{sadinle2017itemwise} show that the full data law is nonparametrically identified under (\ref{eq:NSC}). However, their identification (and estimation) approach requires modeling the full-data joint density. In contrast, we will establish a novel but complementary identification result and variationally-independent parameterization for the probability of missingness $p(R|L)$, which makes inverse probability weighting immediately feasible and obviates the need to model the joint density $p(L)$. 

\begin{figure}[pt]
	\begin{center}
		\begin{tikzpicture}[line width=1.2pt] 
		\node (a) at (-2,2.5) {$L_1$}; 	
		\node (b) at (0,2.5) {$L_2$} 
		edge [<-] node[auto]  {} (a); 
		\node (c) at (2,2.5) {$L_3$}
		edge [<-, bend right=15] node[auto]  {} (a) 
		edge [<-] node[auto]  {} (b);
		\node (d) at (-2,1) {$R_1$}
		edge [<-] node[auto]  {} (b) 	
		edge [<-] node[auto]  {} (c);
		\node (e) at (0,1) {$R_2$}
		edge [<-] node[auto]  {} (a) 	
		edge [<-] node[auto]  {} (c)
		edge [-] node[auto] {} (d);
		\node (f) at (2,1) {$R_3$}
		edge [<-] node[auto]  {} (a) 	
		edge [<-] node[auto]  {} (b)
		edge [-, bend left = 15] node[auto] {} (d)
		edge [-] node[auto] {} (e);
		\end{tikzpicture}
		\caption{A chain graph representation of the no self-censoring independence model, for $K=3$ variables.}
		\label{fig:chain} 	
	\end{center}
\end{figure}
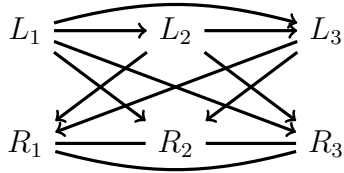

\section{Identification}

In order to fix ideas, suppose that our target parameter of interest is $\beta = \E[B] \equiv \E[b(L)]$ for some known function $b$ of the full data. Later, we consider more general parametric or semiparametric full data models. In order to express this parameter as a function of the observed data, we make use of an odds ratio ($\Odds$) parameterization of the joint density discussed in, e.g., \citet{chen2007semiparametric, chen2010compatibility}. Chen shows how to factorize an arbitrary joint density into a product of variationally-independent components, namely a combination of univariate conditionals and conditional odds ratios. We make use of this factorization to show, first, that the probability of missingness $p(R|L)$ is identified for every missingness pattern. Then, the odds ratio function $\Odds(R,L) \equiv \Odds(R,L;r_0=1,l_0=0) = \frac{p(R|L)}{p(R=1|L)} \frac{p(R=1|L=0)}{p(R|L=0)}$ is also identified.\footnote{We define the odds ratio with fixed reference values $r_0=1$ and $l_0=0$ throughout and suppress the reference value in our notation for brevity.} This enables identification of 
$\beta$. All proofs not in the main body are deferred to the appendix. 

The missingness probability $p(R|L)$ can be expressed using the odds ratio parameterization \citep[Eq.\ 3]{chen2010compatibility}:
\begin{equation} \label{eq:missingnessID}
p(R|L) = \frac{\prod_{i=2}^K \Odds(R_i,(R_1,...,R_{i-1})|(R_{i+1},...,R_K) = 1,L) \prod_{i=1}^K p(R_i|R_{-i}=1,L)}{\sum_{r} \prod_{i=2}^K \Odds(r_i,(r_1,...,r_{i-1})|(R_{i+1},...,R_K) = 1,L) \prod_{i=1}^K p(r_i|R_{-i}=1,L)}
\end{equation}
where $\Odds(R_i,(R_1,...,R_{i-1})|(R_{i+1},...,R_K) = 1,L) = \frac{p(R_i|(R_{i+1},...,R_K) = 1, R_1,...,R_{i-1},L)p(R_i=1|R_{-i} = 1, L)}{p(R_i|R_{-i} = 1, L)p(R_i=1|(R_{i+1},...,R_K) = 1, R_1,...,R_{i-1},L)}$. Under the no self-censoring assumption (\ref{eq:NSC}), each term in this ratio can be written as a function of the observed data. 

\begin{Thm} \label{thm:missingnessID}
	$p(R|L)$ is nonparametrically identified under (\ref{eq:NSC}).
\end{Thm}

To provide some concrete intuition for this result, we illustrate the case where $K=3$.
\begin{equation*}  
\begin{aligned} 
\lefteqn{p(R_1,R_2,R_3|L)}\\ &= p(R_1|R_{-1}=1,L)p(R_2|R_{-2}=1,L)p(R_3|R_{-3}=1,L)\\
&\times \Odds(R_1,R_2|R_3=1,L)\Odds(R_2,R_3|R_1=1,L)\Odds(R_1,R_3|R_2,L) / C(L) \\
&= p(R_1|R_{-1}=1,L)p(R_2|R_{-2}=1,L)p(R_3|R_{-3}=1,L)\\
&\times \Odds(R_1,R_2|R_3=1,L)\Odds(R_2,R_3|R_1=1,L)\Odds(R_1,R_3|R_2=1,L)\Gamma(R_1,R_2,R_3|L) / C(L)\\
&= p(R_1|R_{-1}=1,L_{-1})p(R_2|R_{-2}=1,L_{-2})p(R_3|R_{-3}=1,L_{-3})\\
&\times \Odds(R_1,R_2|R_3=1,L_3)\Odds(R_2,R_3|R_1=1,L_1)\Odds(R_1,R_3|R_2=1,L_2)\Gamma(R_1,R_2,R_3) / C(L)
\end{aligned}
\end{equation*}
where 
\begin{equation*} \begin{aligned}
C(L) =& \sum_{r_1,r_2,r_3} p(r_1|R_{-1}=1,L)p(r_2|R_{-2}=1,L)p(r_3|R_{-3}=1,L)\\
&\times \Odds(r_1,r_2|R_3=1,L)\Odds(r_2,r_3|R_1=1,L)\Odds(r_1,r_3|r_2,L)\\
=& \sum_{r_1,r_2,r_3} p(r_1|R_{-1}=1,L_{-1})p(r_2|R_{-2}=1,L_{-2})p(r_3|R_{-3}=1,L_{-3})\\
&\times \Odds(r_1,r_2|R_3=1,L_3)\Odds(r_2,r_3|R_1=1,L_1)\Odds(r_1,r_3|R_2=1,L_2)\Gamma(r_1,r_2,r_3)\\
\end{aligned}
\end{equation*}
The first equality follows from (\ref{eq:missingnessID}) using the following basic identity implied by the definition of the odds ratio: $\Odds(R_3,(R_1,R_2)|L) = \Odds(R_1,R_3|R_2,L) \Odds(R_2,R_3|R_1=1,L)$. The second equality follows since $\Odds(R_1,R_3|R_2,L)$ can be written equivalently as $\Odds(R_1,R_3|R_2=1,L)\Gamma(R_1,R_2,R_3|L)$ where $\Gamma(R_1,R_2,R_3|L) = \frac{\Odds(R_1,R_3|R_2,L)}{\Odds(R_1,R_3|R_2=1,L)}$, producing an expression in terms of pairwise odds ratios and a 3-way interaction on the odds ratio scale. The final equality follows by assumption (\ref{eq:NSC}). Each pairwise odds ratio
\begin{equation*}\begin{aligned}
\Odds(R_i,R_j|R_k=1,L) &= \frac{p(R_i|R_k=1,R_j,L)p(R_i=1|(R_j,R_k)=1,L)}{p(R_i|(R_j,R_k)=1,L)p(R_i=1|R_k=1,R_j,L)} \\
&= \frac{p(R_i|R_k=1,R_j,L_{-i})p(R_i=1|(R_j,R_k)=1,L_{-i})}{p(R_i|(R_j,R_k)=1,L_{-i})p(R_i=1|R_k=1,R_j,L_{-i})}
\end{aligned}
\end{equation*}
and by symmetry 
\begin{equation*}\begin{aligned}
\Odds(R_i,R_j|R_k=1,L) &= \frac{p(R_j|R_k=1,R_i,L)p(R_j=1|(R_i,R_k)=1,L)}{p(R_j|(R_i,R_k)=1,L)p(R_j=1|R_k=1,R_i,L)} \\
&= \frac{p(R_j|R_k=1,R_i,L_{-j})p(R_j=1|(R_i,R_k)=1,L_{-j})}{p(R_j|(R_i,R_k)=1,L_{-j})p(R_j=1|R_k=1,R_i,L_{-j})}
\end{aligned}
\end{equation*}
therefore $\Odds(R_i,R_j|R_k=1,L) = \Odds(R_i,R_j|R_k=1,L_k)$ a function of only $L_k$. A similar symmetry argument implies that $\Gamma(R_1,R_2,R_3|L)$ is independent of $L$. Finally $p(R_i|R_{-i}=1,L) = p(R_i|R_{-i}=1,L_{-i})$ which establishes that the numerator (and therefore also the normalizing constant) is function of only observed data. We discuss a more abstract derivation which makes use of the chain graph Markov factorization in the appendix.

As an immediate corollary to Theorem \ref{thm:missingnessID}, we have that the odds ratio $\Odds(R,L)$ is identified\footnote{This result has been updated due to an error in the original published version: the last term in $\delta h_i(L_{-i})$ had been omitted. An erratum has been appended to the journal version of this paper to reflect the change. We are grateful to Wang Miao and Yilin Li for bringing the error to our attention. We note that in the $K=2$ case, the ratio $\Odds(\cdot,\cdot|L)/\Odds(\cdot,\cdot|L=0) = 1$ and so the last term vanishes, but that's not true for $K \geq 3$. From the proof of Theorem 1, we know that the product of odds ratios in the factorization can be expressed as a product of every pairwise odds ratio and all $3$-way, ..., $K$-way interaction terms, each of which is a function of only observed data. So these additional terms in the expression above are functions of only observed data.} for any $K$.
\begin{Cor}
	Under assumptions (\ref{eq:NSC}) and (\ref{eq:positivity}), 
	\begin{equation}
	\textup{$\Odds$}(R,L) = \exp \{ (1-R_1) \delta h_1(L_{-1}) + ... +  (1-R_K) \delta h_K(L_{-K}) \}
	\end{equation}
	where $\delta h_i(L_{-i}) = \log( \frac{p(R_i=0|R_{-i}=1,L_{-i})}{p(R_i=1|R_{-i}=1,L_{-i})}) -  \log (\frac{p(R_i=0|R_{-i}=1,L_{-i}=0)}{p(R_i=1|R_{-i}=1,L_{-i}=0)}) + \log( \frac{\textup{OR}(R_i,(R_1,...,R_{i-1})|(R_{i+1},...,R_K) = 1,L_{-i})}{\textup{OR}(R_i,(R_1,...,R_{i-1})|(R_{i+1},...,R_K) = 1,L_{-i}=0)})$ for $i = 1,...,K$ (and the last term is defined to be zero for $i=1$).
\end{Cor}


Finally, we have the following result: 

\begin{Cor}
	Under assumptions (\ref{eq:NSC}) and (\ref{eq:positivity}), $\beta$ is identified by
	\begin{equation}
	\beta = \E[B] = \E \left\{ \sum_r \frac{\E[\textup{$\Odds$}(r,L)b(L)|R=1, L_{(r)}]}{\E[\textup{$\Odds$}(r,L)|R=1, L_{(r)}]} \mathbb{I}(R=r) \right\}. 
	\end{equation}
\end{Cor}
\begin{proof}
	By the bivariate odds ratio factorization \citep[Eq.\ 1]{chen2007semiparametric}, we have that for any measurable function $b$:
	\begin{equation*}
	\E[B|R=r, L_{(r)}] = \frac{\E[\Odds(r,L)b(L)|R=1, L_{(r)}]}{\E[\Odds(r,L)|R=1, L_{(r)}]}
	\end{equation*}
	Therefore,
	\begin{equation*} \begin{aligned}
	\E[B] &= \sum_{l_{(r)}} \sum_r \frac{\E[\Odds(r,L)b(L)|R=1, L_{(r)}]}{\E[\Odds(r,L)|R=1, L_{(r)}]} p(r, l_{(r)})\\
	&= \E \left\{ \sum_r \frac{\E[\Odds(r,L)b(L)|R=1, L_{(r)}]}{\E[\Odds(r,L)|R=1, L_{(r)}]} \mathbb{I}(R=r) \right\}. 
	\end{aligned} \end{equation*}
	(Note that one may replace sums over values of $L$ with integrals as appropriate for continuous components of $L$.)
\end{proof}

Thus $\beta$ can be expressed as a functional of the observed data distribution $p(R,L_{(R)})$ under assumptions (\ref{eq:NSC}) and (\ref{eq:positivity}), provided all terms in the odds ratio factorization of $p(R | L)$ are defined.  A necessary condition for the latter is that all patterns of the form $R_{-i}=1$ for each $i$ (all covariates but one are observed) have support in the data. 

\section{Semiparametric theory}

Suppose that the full data law $p(L;\eta)$ is indexed by an infinite-dimensional parameter $\eta$. Of interest is a finite-dimensional parameter $\beta = \beta(\eta)$. Further, we assume the conditional model $p(R|L; \gamma)$ is indexed by 
an infinite-dimensional parameter $\gamma$. We are interested in deriving the efficient influence function for $\beta$ in the nonparametric no self-censoring model, i.e., the model satisfying (\ref{eq:NSC}) and (\ref{eq:positivity}) but otherwise unrestricted. We denote this model by $\mathcal{M}_{\text{nsc}}$. Among all regular and asymptotically linear (RAL) estimators of $\beta$ in $\mathcal{M}_{\text{nsc}}$, the estimator based on the efficient influence function (that is, solving the efficient IF estimating equation with all nuisance models estimated nonparametrically) achieves, under sufficient regularity conditions, the minimum asymptotic variance and is said to achieve the semiparametric efficiency bound \citep{bickel1993efficient}.

In what follows we define $\pi_j(L) \equiv p(R=r_j | L)$ for missingness patterns $j = 1,...,J$. We reserve $J$ for the complete-case pattern, i.e., $R=r_J=1$. Before presenting the main result for $\mathcal{M}_{\text{nsc}}$, we present the efficient influence function for $\beta$ in a different nonparametric model, not necessarily satisfying the no self-censoring assumption. This influence function is easier to derive and will later on suggest an estimator that is both easier to implement and exhibits an interesting double-robustness property. Let the conditional model $p(R|L)$ be parameterized as $\log \frac{\pi_j(L)}{\pi_J(L)} = h_j(L; \gamma)$. For the next result, we assume that the log odds ratio is a known function of $L$. Specifically, we assume for the moment that $h_j(L; \gamma) = h_{1,j}(L=0;\gamma) + h_{2,j}(L)$ and denote by $\mathcal{M}_{\text{odds}}$ the nonparametric model where $h_{2,j}$ is known for $j=1,...,J$. We use $\phi_{\text{full}}(\beta)$ to denote the full data influence function for $\beta$. (For example, if $\beta = \E[b(L)]$ then $\phi_{\text{full}}(\beta) = b(L) - \beta$.)

\begin{Lma} \label{lma:oddsknown}
	In $\mathcal{M}_{\textup{odds}}$, the efficient influence function for $\beta$ is \begin{align*}
	\phi_{\textup{odds}}(\beta) =& \frac{\mathbb{I}(R=1)}{\pi_J(L)} \phi_{\textup{full}}(\beta) + \sum_{j=1}^{J-1} \mathbb{I}(R=r_j) \E[\phi_{\textup{full}}(\beta) |R=r_j, L_{(r_j)}]\\
	& - \frac{\mathbb{I}(R=1)}{\pi_J(L)} \sum_{j = 1}^{J-1} \pi_j(L) \E[\phi_{\textup{full}}(\beta) |R=r_j, L_{(r_j)}].
	\end{align*} 
\end{Lma}

This result is a version of Theorem 4.1 in \cite{robins2000sensitivity}, though we provide a simple and self-contained proof in the appendix using our notation. Next, we return to the model $\mathcal{M}_{\text{nsc}}$, satisfying the no self-censoring assumption and where the odds ratio is not known a priori. Deriving the influence function for $\beta$ involves two steps: first noticing that the odds ratio is point identified in $\mathcal{M}_{\text{nsc}}$ by the results in the previous section, and second ``adjusting'' the above IF for nonparametric estimation of the odds ratio, subject to the no self-censoring restriction.


\begin{Thm} \label{thm:nonparametricIF}
	In $\mathcal{M}_{\textup{nsc}}$, the efficient influence function for $\beta$ is $\phi_{\textup{nsc}}(\beta) = - \E \left[ \frac{\partial}{\partial \beta} \phi_{\textup{odds}}(\beta) \right]^{-1} \times \left( \phi_{\textup{odds}}(\beta) +  \phi_{\textup{adj}}(\beta) \right)$, with $\phi_{\textup{odds}}(\beta)$ from Lemma \ref{lma:oddsknown} and
	\begin{align*}
	\phi_{\textup{adj}}(\beta) = - \sum_{i=1}^K \E[(1-R_i)|L_{-i}] \frac{\mathbb{I}(R_{-i}=1)}{p(R_{-i}=1|L_{-i})} \left( \frac{R_i}{p(R_i=1|R_{-i},L_{-i})}-1 \right) \\
	\times \frac{  p(R_i=1|R_{-i},L_{-i})}{ p(R_i=0|R_{-i},L_{-i})} \left(1 - \frac{p(R_i=1|R_{(1:i-1)},R_{(i+1:K)}=1,L)}{p(R_i=0|R_{(1:i-1)},R_{(i+1:K)}=1,L)} \right) \E[\Delta(R,L)|R_i=0,L_{-i}] \\
	- \sum_{i=1}^K \E[(1-R_i)|L_{-i}] \frac{\mathbb{I}(R_{(i+1:K)}=1)}{p(R_{(i+1:K)}=1|L_{-i})} \left( \frac{R_i}{p(R_i=1|R_{(1:i-1)},R_{(i+1:K)}=1,L_{-i})}-1 \right) \\
	\times \frac{1}{\textup{OR}(R_i,R_{(1:i-1)}|R_{(i+1:K)} = 1,L_{-i})} \E[\Delta(R,L)|R_i=0,L_{-i}]
	\end{align*}
	with $\Delta(R,L) \equiv \phi_{\textup{full}}(\beta) - \E[\phi_{\textup{full}}(\beta) |R, L_{(R)}].$
\end{Thm}

Therefore\footnote{This result has been updated due to an error in the original published version. The omitted term in Corollary 1 (see footnote above) leads to an additional contribution to this expression from taking the derivative $\nabla_t \log \frac{\Oddst(R_i,(R_1,...,R_{i-1})|(R_{i+1},...,R_K) = 1,L_{-i})}{\Odds(R_i,(R_1,...,R_{i-1})|(R_{i+1},...,R_K) = 1,L_{-i}=0)}$. An erratum has been appended to the journal version of this paper to reflect the change.} the semiparametric efficiency bound in $\mathcal{M}_{\text{nsc}}$ is given by the variance of $\phi_{\text{nsc}}$. 

\section{Double-robustness in settings with always-observed covariates}

In some settings, there may be available an additional set of always-observed covariates $X$. For example, $X$ may consist of baseline measurements (with no missing values) in a longitudinal study with complex patterns of missingness at follow-up times $1,...,K$. We may assume that our fundamental identifying assumptions on the missingness mechanism (\ref{eq:NSC}) and (\ref{eq:positivity}) hold conditional on $X$, i.e.,
\begin{equation} \label{eq:NSCX}
R_i \independent L_i | R_{-i},L_{-i},X
\end{equation}
for $i=1,...,K$ and
\begin{equation} \label{eq:positivityX}
p(R=1|L,X) > \sigma' > 0
\end{equation}
w.p.1 for some constant $\sigma'$. Versions of Lemma \ref{lma:oddsknown} and Theorem \ref{thm:nonparametricIF} hold under these assumptions, where $X$ is added to the conditioning set in all appropriate places. In particular, the conditional odds ratio $\Odds(R,L|X) = \exp \{ (1-R_1) \delta h_1(L_{-1},X) + ... +  (1-R_K) \delta h_K(L_{-K},X) \} $ where $\delta h_i(L_{-i},X) = \log( \frac{p(R_i=0|R_{-i}=1,L_{-i},X)}{p(R_i=1|R_{-i}=1,L_{-i},X)}) -  \log (\frac{p(R_i=0|R_{-i}=1,L_{-i}=0,X)}{p(R_i=1|R_{-i}=1,L_{-i}=0,X)})$, $\pi_j(L,X) = p(R=r_j|L,X)$, and likewise for other terms which were previously considered only functions of $L$. 

It follows immediately that
\begin{align*}
\phi_{\text{odds}}(\beta) =& \frac{\mathbb{I}(R=1)}{\pi_J(L,X)} \phi_{\text{full}}(\beta) + \sum_{j=1}^{J-1} \mathbb{I}(R=r_j) \E[\phi_{\text{full}}(\beta) |R=r_j, L_{(r_j)},X]\\ &- \frac{\mathbb{I}(R=1)}{\pi_J(L,X)} \sum_{j = 1}^{J-1} \pi_j(L,X) \E[\phi_{\text{full}}(\beta) |R=r_j, L_{(r_j)},X]
\end{align*} 
and replacing assumptions (\ref{eq:NSC}) and (\ref{eq:positivity}) with (\ref{eq:NSCX}) and (\ref{eq:positivityX}),  $\phi_{\text{nsc}}(\beta) = - \E \left[ \frac{\partial}{\partial \beta} \phi_{\text{odds}}(\beta) \right]^{-1} \times \left( \phi_{\text{odds}}(\beta) +  \phi_{\text{adj}}(\beta) \right)$ with $\phi_{\text{odds}}(\beta)$ as above and
\begin{align*}
\phi_{\textup{adj}}(\beta) = - \sum_{i=1}^K \E[(1-R_i)|L_{-i},X] \frac{\mathbb{I}(R_{-i}=1)}{p(R_{-i}=1|L_{-i},X)} \left( \frac{R_i}{p(R_i=1|R_{-i},L_{-i},X)}-1 \right) \\
\times \frac{  p(R_i=1|R_{-i},L_{-i},X)}{ p(R_i=0|R_{-i},L_{-i},X)} \left(1 - \frac{p(R_i=1|R_{(1:i-1)},R_{(i+1:K)}=1,L,X)}{p(R_i=0|R_{(1:i-1)},R_{(i+1:K)}=1,L,X)} \right) \E[\Delta(R,L,X)|R_i=0,L_{-i},X] \\
- \sum_{i=1}^K \E[(1-R_i)|L_{-i},X] \frac{\mathbb{I}(R_{(i+1:K)}=1)}{p(R_{(i+1:K)}=1|L_{-i},X)} \left( \frac{R_i}{p(R_i=1|R_{(1:i-1)},R_{(i+1:K)}=1,L_{-i},X)}-1 \right) \\
\times \frac{1}{\textup{OR}(R_i,R_{(1:i-1)}|R_{(i+1:K)} = 1,L_{-i},X)} \E[\Delta(R,L)|R_i=0,L_{-i},X]
\end{align*}
where $\Delta(R,L,X) \equiv \phi_{\text{full}}(\beta) - \E[\phi_{\text{full}}(\beta) |R, L_{(R)},X].$

Interestingly, in the setting with always-observed covariates $X$ we have an estimator that is doubly-robust. Specifically, the estimating function $\phi_{\text{odds}}(\beta)$ above is mean-zero at the true value of $\beta$ in the union model where the odds ratio is correctly specified and for each pattern either the pattern probability $\pi_j(L,X)$ or pattern mixture regression model $\E[B|R=r_j,L_{(r_j)},X]$
is correctly specified, but possibly not both. Let $\mathcal{M}_{\Odds}$ denote the model where $\Odds(R,L|X)$ is correctly specified, $\mathcal{M}_{\pi,j}$ denote the model $p(R=r_j|L=0,X; \psi_j)$ parameterized by $\psi_j$ and $\mathcal{M}_{\text{PM},j}$ denote the model $\E[B|R=r_j,L_{(r_j)}=0,X; \mu_j]$ parameterized by $\mu_j$. Define the union model $\mathcal{M}_{\text{union}} = \cap_j \mathcal{M}_j$, where $\mathcal{M}_j = \left( \mathcal{M}_{\pi,j} \cup \mathcal{M}_{\text{PM},j} \right) \cap \mathcal{M}_{\Odds}$. We use $\psi_{0} = (\psi_{1,0},...,\psi_{J,0})$ and $\mu_{0}=(\mu_{1,0},...,\mu_{J,0})$ to denote the true parameter vectors.

\begin{Thm} \label{thm:DR}
	Let $\phi_{\textup{odds}}(\beta, \textup{$\Odds$}, \mu, \psi)$ as defined above. When $(\psi_0, \mu_0) \in \mathcal{M}_{\textup{union}}$,\\ $\E[\phi_{\textup{odds}}(\beta, \textup{$\Odds$}, \mu_0, \psi_0)] = 0$.
\end{Thm}

Note that the double-robustness property obtained in Theorem \ref{thm:DR} requires that the odds ratio $\Odds(R,L|X)$ is a known function of $L$ and $X$. Since the odds ratio appears in both the pattern mixture regressions and the pattern probabilities $\pi_j(L,X)$, double-robustness in this setting does not protect against \emph{arbitrary} misspecification of either the regression models or pattern probabilities: only components of these models \emph{variationally independent of the odds ratio} may be misspecified without necessarily sacrificing unbiasedness of the estimating equation. Specifically, this implies that at most one of $\pi_j(L=0,X)$ or the regression function $\E[\phi_{\text{full}}|R=r_j,L_{(r_j)}=0,X] = \frac{\E[\Odds(r_j,L|X) \phi_{\text{full}} |R=1, L_{(r_j)}=0, X]}{\E[\Odds(r_j,L|X)|R=1, L_{(r_j)}=0, X]}$ for each pattern may be misspecified. Furthermore, the quantifier over patterns means that for some patterns one may correctly specify only the pattern probability and for other patterns one may only specify the pattern mixture regression without sacrificing unbiasedness.

It is instructive to contrast our double-robustness result with another recently proposed doubly-robust estimator for a MNAR model, the ``discrete choice model'' (DCM) estimator in \cite{tchetgen2018discrete}. The model for the missingness mechanism introduced in that paper is motivated by some behavioral assumptions underlying observed patterns of nonresponse in the data. In that model (which is neither properly a superset nor subset of our model $\mathcal{M}_{\text{nsc}}$), \cite{tchetgen2018discrete} propose an estimator which is doubly-robust in the sense of requiring that either the pattern mixture regression or missingness mechanism is correctly specified. In our case we also require that the odds ratio is correctly specified for each pattern. However, in the special setting where the data only contains complete cases and every ``leave-one-out'' pattern (i.e., patterns where $L_i$ is missing but $L_{-i}$ is observed, for each $i$), then the no self-censoring and DCM models coincide. That is, if every missingness pattern besides the complete case and ``leave-one-out'' patterns have zero probability, then assumption (\ref{eq:NSC}) and the DCM independence assumption place exactly the same restriction on the missingness mechanism. \cite{tchetgen2018discrete} express the DCM assumption as $L_{(-r)} | R=r, L_{(r)} \sim L_{(-r)} | R=1, L_{(r)}$ for all $r \neq 1$ where $L_{(-r)}$ denotes the unobserved subvector of $L$ when $R=r$ (see also \cite{little1993pattern}). With only complete cases and ``leave-one-out'' patterns this is simplifies to $L_i |R_i=0,R_{-i}=1,L_{-i} \sim L_i |R=1,L_{-i}$. In the same setting, the no self-censoring assumption $L_i|R_i=0,R_{-i},L_{-i}  \sim L_i|R_i=1,R_{-i},L_{-i}$ amounts to $L_i|R_i=0,R_{-i}=1,L_{-i}  \sim L_i|R=1,L_{-i}$. Therefore here the models coincide and thus have the same influence function. Moreover, in this restricted pattern setting the odds ratio function $\Odds(R,L)$ is only involved in the missingness mechanism but not in the pattern mixture regression models, so an estimator based on $\phi_{\text{odds}}$ recovers the same double-robustness property in \cite{tchetgen2018discrete}.

In the next section, we propose an estimator for $\beta$ based on parametric specification of each component of the estimating equation. When we use a doubly robust estimator for the odds ratio components, the resulting estimator for $\beta$ will be doubly-robust in the sense just described.

\section{The proposed estimator}

In applications, it is common to specify parametric models for the odds ratio as well as the pattern probabilities $\pi_j(L,X)$ and the pattern mixture regressions, particularly if $L$ has more than two continuous components. A convenient choice may be to assume a logistic model for $R_i$ in terms of two components, one of which appears only in the odds ratio and both of which appear in the pattern probabilities $\pi_j(L,X)$. For example:
\begin{equation} \label{eq:logit}
\logit p(R_i=1|R_{-i}=1,L_{-i}, X; \psi_i) = \delta h_i(L_{-i},X; \psi_{i,LX}) + h_{i,X}(X; \psi_{i,X})
\end{equation}
where $\psi_i = (\psi_{i,LX}, \psi_{i,X})'$. Note that $\delta h_i(L_{-i},X; \psi_{i,LX})$ is a component of the odds ratio and so $\delta h_i(L_{-i},X; \psi_{i,LX})=0$ when $L_{-i}=0$. The second component $h_{i,X}(X; \psi_{i,X})$ is a function of $X$ which is needed for the $\pi_j(L,X)$ but does not appear in the odds ratio. The parameterized odds ratio is then $\Odds(\cdot; \psi_{LX})$ where $\psi_{LX} \equiv ( \psi_{1,LX},...,\psi_{K,LX} )'$ and we write $\widehat{\Odds}(\cdot) \equiv \Odds(\cdot; \hat{\psi}_{LX})$ for the estimated odds ratio. Also let $\psi_X \equiv ( \psi_{1,X},...,\psi_{K,X} )'$, $\psi \equiv ( \psi_{X},\psi_{LX} )'$, and denote each estimated pattern probability by $\pi_j(L,X; \hat{\psi})$. We use $\mu \equiv (\mu_1,\mu_2)'$ to parameterize the pattern mixture regression functions.


We propose a straightforward augmented IPW (AIPW) estimator for $\beta$, where the augmentation term incorporates information from all the missingness patterns. Denote by $\hat{\beta}_{AIPW}$ the solution to:
\begin{align*}
\mathbb{P}_n \left( \frac{\mathbb{I}(R=1)}{\pi_J(L,X; \hat{\psi})} \phi_{\text{full}}(\beta) + \sum_{j=1}^{J-1} \mathbb{I}(R=r_j) \frac{\E[\widehat{\Odds}(r_j,L|X) \phi_{\text{full}}(\beta) |R=1, L_{(r_j)}, X; \hat{\mu}_1]}{\E[\widehat{\Odds}(r_j,L|X)|R=1, L_{(r_j)}, X; \hat{\mu}_2]} \right. \\
\left. - \frac{\mathbb{I}(R=1)}{\pi_J(L,X; \hat{\psi})} \sum_{j = 1}^{J-1} \pi_j(L,X; \hat{\psi}) \frac{\E[\widehat{\Odds}(r_j,L|X) \phi_{\text{full}}(\beta) |R=1, L_{(r_j)},X; \hat{\mu}_1]}{\E[\widehat{\Odds}(r_j,L|X)|R=1, L_{(r_j)},X; \hat{\mu}_2]} \right) = 0
\end{align*}
where $\mathbb{P}_n$ denotes the sample average. This (empirically) solves the estimating equation $\E[\phi_{\text{odds}}(\beta;\hat{\psi}_X, \hat{\mu}, \widehat{\Odds})]=0$ with estimators for all nuisance functions plugged-in. In order to achieve double-robustness, we must use an estimator of the odds ratio which is consistent in the union model. One such estimator, proposed in \cite{tchetgen2010doubly} and \cite{tan2019doubly} for a full data semiparametric problem, can readily be adapted to the missing data setting. The doubly-robust estimator is $\widehat{\Odds}_{dr} \equiv \Odds(\cdot; \hat{\psi}_{LX})$ where each $\hat{\psi}_{i,LX}$ solves the estimating equation $\mathbb{P}_n (r(\psi_{i,LX}; \hat{\psi}_{i,X}, \hat{\mu} )) = 0$ and
\begin{align*}
\lefteqn{r( \psi_{i,LX} ; \hat{\psi}_{i,X}, \hat{\mu} ) =}\\ 
& & \left( R_i - \mbox{expit}(h_{i,X}(X; \hat{\psi}_{i,X})) \right) \left( \frac{\partial \delta h_i(L_{-i},X; \psi_{i,LX}) }{\partial \psi_{i,LX}} - \E \left[\frac{\partial \delta h_i(L_{-i},X; \psi_{i,LX})}{\partial \psi_{i,LX}}|X;\hat{\mu} \right] \right)\\
& & \times \exp\left\{ -\delta h_i(L_{-i},X; \psi_{i,LX}) \right\} \mathbb{I}(R_{-i}=1) 
\end{align*}
The resulting estimator $\hat{\psi}_{LX}$ is consistent if either $p(R_i=1 | R_{-i}=1,L_{-i}=0,X)$ or $\E \left[\frac{\partial \delta h_i(L_{-i},X; \psi_{i,LX})}{\partial \psi_{i,LX}}|R=1, X \right]$ for each $i = 1,...,K$ is correctly specified. In practice, a common choice of functional form for $\delta h_i(L_{-i},X; \psi_{i,LX})$ is linear in $L_{-i}$ (as is assumed in \cite{tan2019doubly}) so the derivatives reduce to $L_{-i}$.

Let $V(\cdot)$ be the vector of stacked estimating equations for parameters $\hat{\beta}_{AIPW}$, $\hat{\psi}_X$, 
$\hat{\mu}$, 
and $\widehat{\Odds}_{dr}$ using the doubly-robust odds ratio estimator above. Let $\Omega = (\beta, \psi_X, \mu, \Odds)'$ be the combined set of parameters. Ultimately our procedure solves the estimating equation $\mathbb{P}_n \left[ V(\widehat{\Omega}) \right] = 0$.

\begin{Thm} \label{thm:estimator}
	In the union model $\mathcal{M}_{\textup{union}}$, $\widehat{\Omega}$ is consistent and asymptotically normal with influence function $\E \left[ \frac{ \partial V(\Omega)}{ \partial \Omega } \right]^{-1} V(\Omega)$. 
\end{Thm}


Thus, the proposed estimator has the benefit of being simple to implement and doubly-robust with respect to the always-observed covariates. We note that implementing an estimator based on the nonparametric IF $\phi_{\text{nsc}}$ would be asymptotically more efficient when all parametric models are correctly specified, but also considerably more complicated to specify correctly because of the $p(R_{-i}=1|L_{-i},X)$ terms; each of these would require correctly specifying a joint distribution and then marginalizing by integrating or summing $L_i$ (that is, calculating $p(L_{-i}) = \int_{l_i}p(L)dL_i$). As far as we are able to determine, although locally semiparametric efficient in $\mathcal{M}_{\text{nsc}}$, an estimator based on $\phi_{\text{nsc}}$ fails to be doubly-robust as the entire missing data mechanism must be correctly specified. 
So, in the interests of ease-of-implementation and double-robustness, we propose the simpler $\hat{\beta}_{AIPW}$ and expore its performance with simulations in the next section.

\section{Simulation study}

In our simulations, we specifically consider the case where $K=3$ and $\beta \equiv \E[L_3]$, i.e., we are simply interested in the marginal mean of the ``outcome'' variable, $L_3$. First we examine the setting where all variables are sometimes missing ($X$ is empty) and then the setting with an always-observed vector $X = (X_1,X_2)$. In both cases, we sample $(R,L,X)$ from a conditional Gaussian chain graph model satisfying the no self-censoring assumption; that is, missingness paterns $R=r$ are sampled according to a multinomial distribution and $L,X|R=r$ is normal $N(\mu_0(r),\Sigma_0)$. With this parametric data-generating process, imposing the no self-censoring assumption amounts to setting certain mean parameters (interaction terms) to zero; see \citet[p.~119-120]{hojsgaard2012graphical}. The precise parameter values chosen for the simulation study are detailed in the appendix. We also carried out versions of the following simulations with binary data; those results are deferred to the appendix.

\subsection{Setting 1}

In Figure \ref{sim1}, we compare the performance of our AIPW estimator against a popular multiple imputation method for missing data: multivariate imputation by chained equations or MICE \citep{buuren2010mice}. Though MICE uses a series of flexible models to impute missing values based on covariate information, the consistency of this imputation procedure depends on the assumption that missing data mechanism is MAR. Here we have generated data that are MNAR, satisfying assumption (\ref{eq:NSC}), and therefore, as expected, MICE performs poorly despite the quite simple parametric model for the full data. 
With all nuisance models correctly specified, the proposed AIPW estimator is seen to be unbiased.

\begin{figure}[t!] 
	\begin{center}
		\includegraphics[scale=.50]{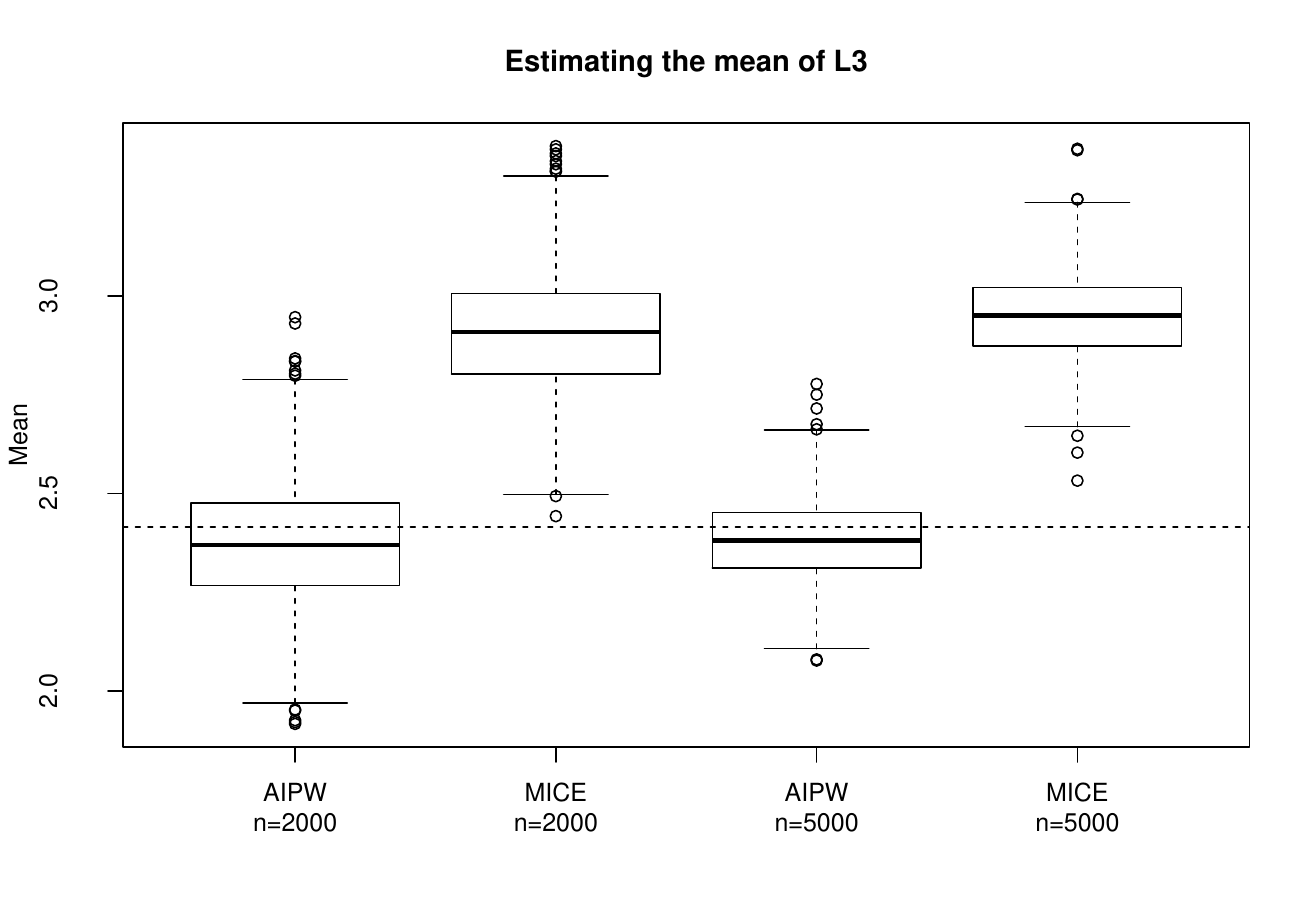} 
	\end{center}
	\caption{Estimates of $\E[L_3]$ in the first simulation setting. Boxplots are calculated from 1000 trials at sample sizes $n=2000$ and $n=5000$. The horizontal dashed line indicates the true value, which is 2.415593. MICE estimates are clearly biased upwards. 
		AIPW estimates, as expected, concentrate around the true value.}
	\label{sim1}
\end{figure}

\subsection{Setting 2}

In the second setting, with always observed vector $X = (X_1,X_2)$, we explored the double-robustness property. From the conditional Gaussian form of the data-generating process, we know that the outcome regressions are correctly specified as log-linear functions of $L_{(r)}$ and $X$. Likewise, the logit probability of each missingness indicator $R_i$ in eq.\ (\ref{eq:logit}) is a linear function of $L_{-i}$ and $X$. To illustrate our double-robustness result, these nuisance models were misspecified by replacing $(X_1,X_2)$ with 
$(\log(\frac{1}{X_1}+ \frac{1}{X_2}), \sqrt{X_1X_2})$
in all outcome regressions and missingness probabilities. We used the doubly-robust estimator $\widehat{\Odds}_{dr}$ for the odds ratio. Results for 1000 trials with $n=5000$ samples are shown in Table 1. The AIPW estimator is seen to be unbiased when either the outcome regressions or missingness probabilities are misspecified, but not both.

\begin{table}[t!]
	\begin{center}
		\begin{tabular}{||c | c | c | c | c ||} 
			\hline
			\textbf{Setting} & \textbf{Bias} & \textbf{Percent Bias} & \textbf{MSE} & \textbf{Var} \\
			\hline\hline
			Outcome Reg. & -0.16 & 0.98 & 0.14 & 0.12 \\
			Misspec. & & & &\\
			\hline
			Missingness Prob. & -0.16 & 0.96 & 0.039 & 0.013 \\
			Misspec. & & & &\\
			\hline
			Both & 0.96 & 5.75 & 7.68 & 6.77 \\
			Misspec. & & & &\\
			\hline
			Both & -0.068 & 0.41 & 0.020 & 0.016 \\
			Correct & & & &\\
			\hline
		\end{tabular}
		\caption{Simulation results illustrating the double-robustness property. The true value of $\E[L_3]$ is 16.71512 and sample size is $n=5000$. Bias, percent bias, mean squared error, and variance are calculated over 1000 trials.}
	\end{center}
\end{table}


\section{An application to HIV data}

We applied the proposed AIPW estimator to data from an observational study of HIV-positive mothers in Botswana. Specifically, we are interested in the relationship between continuing highly active antiretroviral therapy (HAART) during pregnancy and adverse birth outcomes, such as preterm delivery. The full data set, abstracted from 6 sites in Botswana, is described in detail in \cite{chen2012highly}. Following \cite{tchetgen2018discrete}, \cite{sun2018inverse}, and \cite{shpitser2016consistent} we focus on HIV-positive women ($n=9711$) and 3 variables: HAART exposure during pregnancy, preterm delivery, and an indicator of low CD4$^+$ count (less than 200 $\mu L$). The question of interest is whether HAART continuation (68.9\% missing) is associated with premature delivery (6.7\% missing), satfified by CD4$^+$ count (53.4\% missing). The analysis is complicated by a small number of complete cases (10.5\%) and nonmonotone patterns of missingness. 

By way of estimating the association of interest and as further illustration of the proposed approach, we first obtain an estimate of the joint distribution of the variables of interest. In Table 2, we compare the joint distribution as estimated by our proposed AIPW estimator with complete-case analysis as well as MICE. Substantial differences between our AIPW procedure and complete case analysis are evident, for example in the probability of observing HAART continutation, no low CD4$^+$ count, and no preterm delivery (third column). Differences with MICE are less pronounced. From the estimated joint distribution, we can obtain an estimate of the odds ratios between HAART and preterm delivery at both levels of low CD4$^+$ count: these are 2.72 (95\% CI: 1.26, 5.53) and 1.24 (95\% CI: 0.91, 1.73) for low CD4$^+$ count and moderate or high CD4$^+$ count, respectively. Ninety-five percent confidence intervals in parentheses are computed by bootstrap (percentile) over 1000 subsamples. Using the same procedure, IPW (i.e., our estimator with pattern mixture regression models set to zero) produces estimates 2.51 (95\% CI: 1.16, 4.95) and 1.15 (95\% CI: 0.83, 1.63) respectively, MICE produces estimates 1.60 (95\% CI: 1.09, 2.90) and 1.52 (95\% CI: 0.98, 1.70), while compete-case analysis produces 2.16 (95\% CI: 0.98, 4.41) and 1.00 (95\% CI: 0.71, 1.42) respectively. It is interesting to compare our AIPW estimates to the results of \cite{tchetgen2018discrete}, who analyze the same data with the aforementioned discrete choice model estimator. Their analysis is based on a different MNAR assumption than the one considered here, as discussed in Section 5. They report an estimated odds ratio association of 1.158 (95\% CI: 0.869, 1.560) under a main effect-only logistic regression of preterm delivery on HAART continuation and CD4$^+$ count. Similarly, \cite{shpitser2016consistent} also reports an estimated odds ratio association of 1.032 (95\% CI: 0.670, 1.394) using his pseudolikelihood-based IPW estimator under the no self-censoring assumption. Neither analysis by Tchetgen Tchetgen et al.\ and Shpitser was able to detect a significant association between preterm delivery and HAART continuation conditional on CD4$^+$ count; in contrast, the proposed AIPW estimator detected a significant association 
for mothers with low CD4$^+$ count. This underscores how the assumed missingness model may have quite important implications for the substantive scientific conclusions or policy recommendations supported by the data.

\begin{table}[t!]
	\begin{center}
		\begin{tabular}{||c | c | c | c | c | c | c | c | c||} 
			\hline
			H,C,P & 1,1,1 & 1,1,0 & 1,0,0 & 1,0,1 & 0,1,1 & 0,1,0 & 0,0,0 & 0,0,1\\
			\hline\hline
			AIPW & 0.0122 & 0.0268 & 0.5206 & 0.1722 & 0.0091 & 0.0542 & 0.1618 & 0.0430 \\
			\hline
			MICE & 0.0164 & 0.0355 & 0.5156 & 0.1705 & 0.0140 & 0.0483 & 0.1641 & 0.0357\\
			\hline
			Complete Cases & 0.0137 & 0.0342 & 0.3320 & 0.0979 & 0.0333 & 0.1802 & 0.2380 & 0.0705\\
			
			\hline
		\end{tabular}
		\caption{Estimated joint distribution of HAART continuation during pregnancy (H), low CD4$^+$ count (C), and preterm delivery (P) in HIV-infected women in Botswana: comparing complete case analysis and MICE with estimation by AIPW.}
	\end{center}
\end{table}

\section{Discussion}

We have introduced a practical and straightforward-to-implement AIPW estimator for functions of data missing-not-at-random, under the ``no self-censoring'' or ``itemwise conditionally indendent nonresponse'' assumption, which places no restrictions on the observed data. Our estimator improves on the efficiency and flexibility of previously proposed estimators \citep{shpitser2016consistent, sadinle2017itemwise} and when a subset of covariates are always observed, enjoys a certain double-robustness property (provided the odds ratio function encoding the association between $R$ and $L$ is correctly specified). We demonstrated in simulations that the estimator is an attractive alternative to popular multiple imputation procedures when the missing data mechanism is MNAR. Our analysis of HIV data from Botswana demonstrates that the proposed estimator can potentially make an important practical difference in applied problems with acute missingness. 

\section{Appendix}

\subsection{Parameter settings for simulation study}

For $K=3$ missing variables, missingness patterns $(1,0,0), (0,1,0), (1,1,0), (0,0,1), (1,0,1),$ $(0,1,1),  (1,1,1),$ and $(0,0,0)$ were sampled with probabilities $0.169, 0.153, 0.136, 0.119, 0.102,$ $0.085, 0.169,$ and $0.068$ respectively (rounded to the third decimal place). In the setting with $X$ empty:
\begin{equation*}
\Sigma_0 = \begin{pmatrix}
4.4 & 1.3 & -2.8\\
1.3 & 3.2 & 1.3\\
-2.8 & 1.3 & 3.5
\end{pmatrix}
\end{equation*}
and $\mu_0(r) = \Sigma_0 h(r)$ where
\begin{align*}
h(1,0,0) &= (1.4,1.6,0.9)'\\
h(0,1,0) &= (1.9,1.1,1.4)'\\
h(1,1,0) &= (1.9,1.6,0.2)'\\
h(0,0,1) &= (0.5,1.9,2.1)'\\
h(1,0,1) &= (0.5,2.4,0.9)'\\
h(0,1,1) &= (1.0,1.9,1.4)'\\
h(1,1,1) &= (1.0,2.4,0.2)'\\
h(0,0,0) &= (1.4,1.1,2.1)'
\end{align*}
In the setting with $X=(X_1,X_2)$:
\begin{equation*}
\Sigma_0 = \begin{pmatrix}
3.88 & 2.66 & 1.24 & 1.60 & 0.30\\
2.66 & 3.24 & 2.66 & 2.26 & 0.96\\
1.24 & 2.66 & 3.70 & 1.64 & 0.64\\
1.60 & 2.26 & 1.64 & 2.00 & 0.60\\
0.30 & 0.96 & 0.64 & 0.60 & 1.70
\end{pmatrix}
\end{equation*}
and
\begin{align*}
h(1,0,0) &= (1.4,1.6,0.9,2.05,4.15)'\\
h(0,1,0) &= (1.9,1.1,1.4,2.6,2.6)'\\
h(1,1,0) &= (1.9,1.6,0.2,2.6,3.7)'\\
h(0,0,1) &= (0.5,1.9,2.1,3.0,2.7)'\\
h(1,0,1) &= (0.5,2.4,0.9,2.95,3.75)'\\
h(0,1,1) &= (1.0,1.9,1.4,3.8,2.1)'\\
h(1,1,1) &= (1.0,2.4,0.2,3.45,3.45)'\\
h(0,0,0) &= (1.4,1.1,2.1,1.75,3.35)'
\end{align*}
It is straightforward to confirm with these parameter settings that the no self-censoring assumption is satisfied in the data generated.

\subsection{Additional simulation results}

Here we report the results of additional simulation experiments with binary data, similar to Setting 1 and Setting 2 in the main paper. Again, $K=3$ but in this case the target parameter is $\E[L_1\times L_2\times L_3]$ (the probability that all $3$ covariates take the value $1$). Binary data was generated to satisfy the no self-censoring assumption with and without always-observed covariates $X$ by sampling from the joint distribution $p(R,L,X)$. This joint was parameterized using a version of the aforementioned odds ratio factorization \citep{chen2010compatibility}. The $\delta h_i(L_{-i},X)$ terms in $\Odds(R,(L,X))$ were specified as follows:

\begin{align*}
\delta h_1(L_{-1},X) &= -0.8(1-L_2) + 0.6L_3 + 0.5 X_1 + 0.7 X_2 + 0.7 X_1X_2\\
\delta h_2(L_{-2},X) &= -0.8(1-L_1) + 0.7L_3 - 0.7 (1-X_1) - 0.5 X_2 + 0.7 X_1X_2\\
\delta h_3(L_{-3},X) &= 0.5L_1 + 0.5(1-L_2) + 0.2 X_1 - 0.9 X_2 + 0.7 X_1X_2
\end{align*}
with $X_1,X_2$ absent whenever $X$ is empty.
 
In the setting with no always-observed covariates, the results are displayed in Figure \ref{sim1b}. The AIPW estimator concentrates around the true value. In this binary setting MICE also does well. We note that complete case analysis here does very poorly: the complete case estimator is severely biased with a mean of approximately $0.167$ while the true value is $0.321$.  

\begin{figure}[h!] 
	\begin{center}
		\includegraphics[scale=.50]{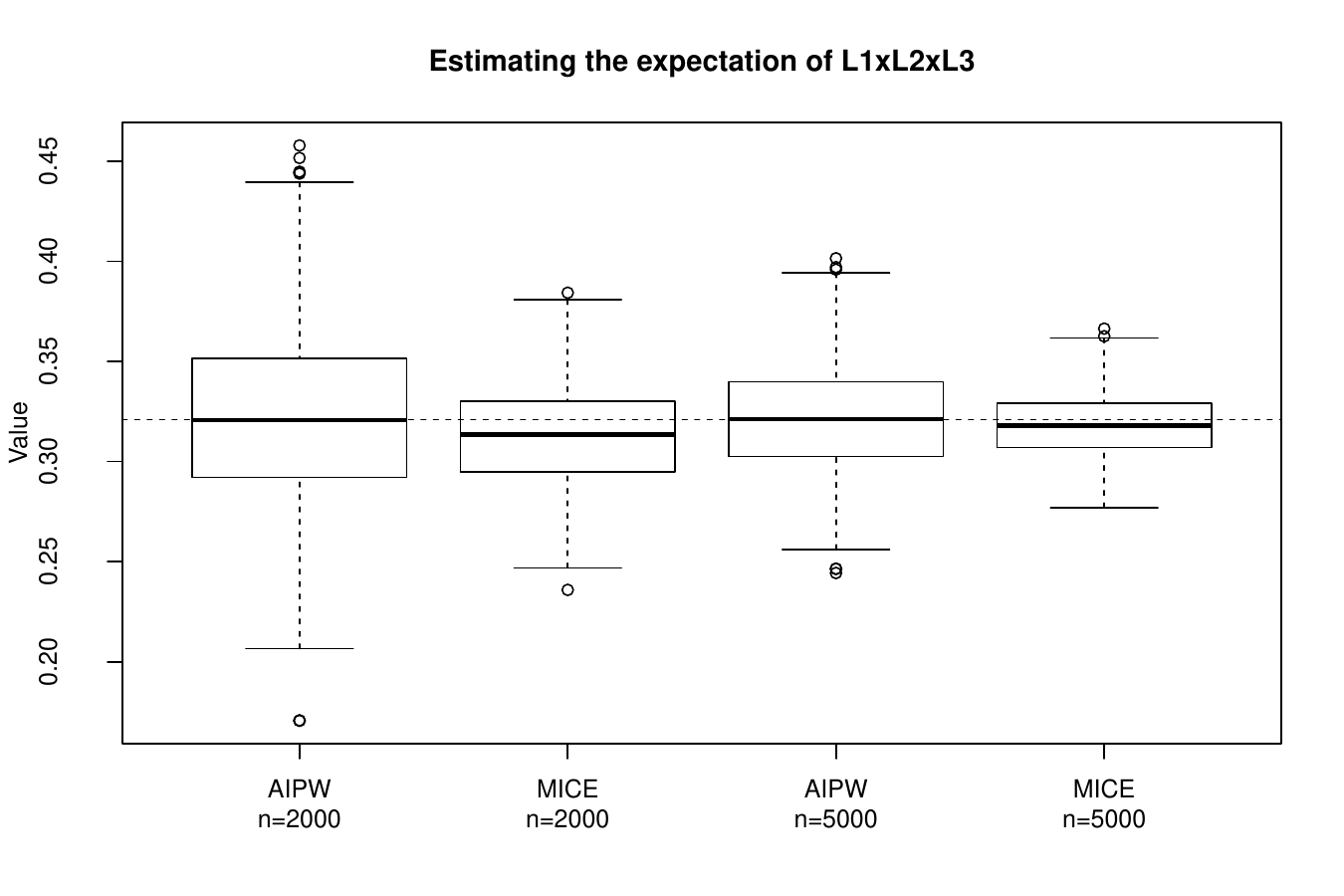} 
	\end{center}
	\caption{Estimates of $\E[L_1 \times L_2 \times L_3]$ in the binary simulation setting. Boxplots are calculated from 1000 trials at sample sizes $n=2000$ and $n=5000$. The horizontal dashed line indicates the true value, which is $0.3209396$.}
	\label{sim1b}
\end{figure}

Results for the setting with always-observed covariates $X$ are reported in Table 3. We see the same general pattern as in the case with continuous variables: the proposed estimator is (approximately) unbiased when either the outcome regressions or missingness probabilities are misspecified, but not both.

\begin{table}[h!]
	\begin{center}
		\begin{tabular}{||c | c | c | c | c ||} 
			\hline
			\textbf{Setting} & \textbf{Bias} & \textbf{Percent Bias} & \textbf{MSE} & \textbf{Var} \\
			\hline\hline
			Outcome Reg. & -0.020 & -7.16 & 0.0039 & 0.0035 \\
			Misspec. & & & &\\
			\hline
			Missingness Prob. & -0.020 & -7.09 & 0.0037 & 0.0033 \\
			Misspec. & & & &\\
			\hline
			Both & -0.037 & -13.25 & 0.0043 & 0.0029 \\
			Misspec. & & & &\\
			\hline
			Both & -0.0045 & -1.6 & 0.0034 & 0.0034 \\
			Correct & & & &\\
			\hline
		\end{tabular}
		\caption{Simulation results in the binary setting. The true value of $\E[L_1\times L_2\times L_3]$ is 0.280638 and sample size is $n=5000$. Bias, percent bias, mean squared error, and variance are calculated over 1000 trials.}
	\end{center}
\end{table}

\subsection{Estimation of the pairwise odds ratios}

The pairwise odds ratios in (\ref{eq:missingnessID}) for $K=3$ can be parameterized as $\theta_1(L_1) = \Odds(R_2,R_3|R_1=1,L_1)$, $\theta_2(L_2) = \Odds(R_1,R_3|R_2=1,L_2)$, $\theta_3(L_3) = \Odds(R_1,R_2|R_3=1,L_3)$, and $\theta_4 = \Gamma(R_1,R_2,R_3)$. Thus, for $K=3$, we have the following four estimating equations.


The estimating equation for $\theta_1$ is:
\begin{equation*}\begin{aligned}
U(\theta_1) &= g_1(L_1) \left[ \frac{R_1R_2R_3}{p(R=1|L)}p(R_1=1,R_2=0,R_3=0|L) - R_1(1-R_2)(1-R_3) \right]\\
&= g_1(L_1) \left[ \frac{R_1R_2R_3}{\prod_{i=1}^{3}p(R_i=1|R_{-i},L_{-i})}\exp(\theta_1(L_1)) p(R_1=1|R_{-1}=1,L_{-1})p(R_2=0|R_{-2}=1,L_{-2}) \right.\\
&\left. \times p(R_3=0|R_{-3}=1,L_{-3}) - R_1(1-R_2)(1-R_3) \Big] \right.
\end{aligned} \end{equation*}
for arbitrary function $g_1$. Similarly,
\begin{equation*}\begin{aligned}
U(\theta_2) &= g_2(L_2) \left[ \frac{R_1R_2R_3}{p(R=1|L)} p(R_1=0,R_2=1,R_3=0|L) - (1-R_1)R_2(1-R_3) \right] \\
&=  g_2(L_2) \left[ \frac{R_1R_2R_3}{\prod_{i=1}^{3}p(R_i=1|R_{-i},L_{-i})} \exp(\theta_2(L_2)) p(R_1=0|R_{-1}=1,L_{-1})p(R_2=1|R_{-2}=1,L_{-2}) \right.\\
&\left. \times p(R_3=0|R_{-3}=1,L_{-3}) - (1-R_1)R_2(1-R_3) \Big] \right.
\end{aligned} \end{equation*}
\begin{equation*}\begin{aligned}
U(\theta_3) &= g_3(L_3) \left[ \frac{R_1R_2R_3}{p(R=1|L)} p(R_1=0,R_2=0,R_3=1|L) - (1-R_1)(1-R_2)R_3 \right]\\
&= g_3(L_3) \left[ \frac{R_1R_2R_3}{\prod_{i=1}^{3}p(R_i=1|R_{-i},L_{-i})} \exp(\theta_3(L_3)) p(R_1=0|R_{-1}=1,L_{-1})p(R_2=0|R_{-2}=1,L_{-2}) \right. \\
& \left. \times p(R_3=1|R_{-3}=1,L_{-3}) - (1-R_1)(1-R_2)R_3 \Big] \right.
\end{aligned} \end{equation*}
and
\begin{equation*} \begin{aligned}
U(\theta_4) &= \frac{R_1R_2R_3}{p(R=1|L)} p(R_1=0,R_2=0,R_3=0|L) - (1-R_1)(1-R_2)(1-R_3)\\
&= \frac{R_1R_2R_3}{\prod_{i=1}^{3}p(R_i=1|R_{-i},L_{-i})} \exp(\theta_1(L_1) 
+ \theta_2(L_2) 
+ \theta_3(L_3) + \theta_4) \prod_{i=1}^3 p(R_i=0|R_{-i}=1,L_{-i})\\
&- (1-R_1)(1-R_2)(1-R_3)
\end{aligned} \end{equation*}
for arbitrary functions $g_2,g_3$. It is easy to see that $\E[U(\theta_i)]=0$ for each $i$. For arbitrary $K$, the number of parameters and number of estimating equations grows rapidly, since one must consider $K$-way interactions, all $K$-1-way interactions, and so on. 

In the setting with always-observed covariates $X$ and pairwise odds ratios such as $\Odds(R_i,R_j|R_k=1,X)$, each estimating equation above becomes a system of simultaneous estimating equations with dimension depending on the size of $X$. For example, corresponding to $\theta_1$ we solve $\E[U(\theta_1)]=0$ where:
\begin{equation*}\begin{aligned}
U(\theta_1) &= 
g_1(L_1,X) \left[ \frac{R_1R_2R_3}{\prod_{i=1}^{3}p(R_i=1|R_{-i},L_{-i},X)}
\exp(\theta_1(L_1,X)) \right. \\
& \times p(R_1=1|R_{-1}=1,L_{-1},X)p(R_2=0|R_{-2}=1,L_{-2},X)p(R_3=0|R_{-3}=1,L_{-3},X)\\
&- \left. R_1(1-R_2)(1-R_3) \Big] \right.
\end{aligned} \end{equation*}
One may proceed similarly for the other pairwise odds ratios.

\subsection{Proofs and derivations}

We begin by discussion the chain graph Markov factorization in order to connect Theorem~\ref{thm:missingnessID} to graphical concepts.

To state the Markov factorization we need several graphical definitions. A block $B \subseteq V$ is defined as a maximal set of vertices such that every vertex pair in the set is connected by an undirected path. The set of blocks in a graph $\mathcal{G}$ is denoted $\mathcal{B(G)}$. A clique $C$ is defined as a maximal set of vertices that are pairwise connected by undirected edges and the set of cliques in $\mathcal{G}$ is denoted $\mathcal{C(G)}$. Let $\text{pa}_i(\mathcal{G})$ denote the parents of vertex $V_i$ in $\mathcal{G}$ and we apply this definition disjunctively to sets so $\text{pa}_B(\mathcal{G})$ is the union of the parents of all vertices in $B$. For a block $B$ we denote the boundary by $\text{bd}_B(\mathcal{G}) = B \cup \text{pa}_B(\mathcal{G})$. We define the augmented graph $\mathcal{G}^a$ to be an undirected graph constructed from $\mathcal{G}$ by replacing all directed edges with undirected edges and connecting all vertices in $\text{pa}_B(\mathcal{G})$ for every block $B$ in $\mathcal{G}$ by undirected edges. $\mathcal{G}_{\text{bd}_B(\mathcal{G})}^a$ is the augmented graph constructed from the induced subgraph over $\text{bd}_B(\mathcal{G})$. Let $\mathcal{C}_B^*(\mathcal{G}) = \{C \in \mathcal{C}(\mathcal{G}_{\text{bd}_B(\mathcal{G})}^a) : C \not \subseteq \text{pa}_B(\mathcal{G}) \}$. The chain graph Markov factorization consists of two parts:
\begin{equation}
p(V) = \prod_{B \in \mathcal{B(G)}} p(B | \text{pa}_B(\mathcal{G}))
\end{equation}
and
\begin{equation}
p(B | \text{pa}_B(\mathcal{G})) = \frac{1}{Z(\text{pa}_B(\mathcal{G}))} \prod_{C \in \mathcal{C}_B^*(\mathcal{G})} \phi_C(C)
\label{eq:innerMarkov}
\end{equation}
for each block $B$, where $\phi_C(C)$ is a ``clique potential function'' depending only on $C$ and $Z(\text{pa}_B(\mathcal{G}))$ is a normalizing factor \citep[cf.][p.~53]{lauritzen1996graphical}.

For the chain graph in Figure \ref{fig:chain} with $K=3$ there are four blocks: $R=(R_1,R_2,R_3)$ is a block and each of $L_1,L_2,L_3$ is a singleton block. The ``inner'' factorization (\ref{eq:innerMarkov}) is only non-trivial for the block $R=(R_1,R_2,R_3)$ where this amounts to:
\begin{equation}
p(R | L) = \frac{1}{Z(L)} \prod_{C \in \mathcal{C}_R^*(\mathcal{G})} \phi_C(C).
\end{equation}
The cliques in $\mathcal{C}_R^*(\mathcal{G})$ are $C_1 = (R_1,L_2,L_3)$, $C_2 = (R_2,L_1,L_3)$, $C_3 = (R_3,L_1,L_2)$, $C_4 = (R_1,R_2,L_3)$, $C_5 = (R_2,R_3,L_1)$, $C_6 = (R_1,R_3,L_2)$, $C_7 = (R_1,R_2,R_3)$. Note that even though $(L_1,L_2,L_3)$ is a clique, this does not appear in the set $\mathcal{C}_R^*(\mathcal{G})$ since it is equal to $\text{pa}_R(\mathcal{G})$. So, there can be no potential function depending on $L=(L_1,L_2,L_3)$. Comparing terms to the odds ratio factorization for $K=3$ in Section 3, we see that cliques $C_1,C_2,C_3$ correspond to the conditional probability terms $p(R_i|R_{-i}=1,L_{-i})$ $i=1,..,3$ and cliques $C_4,C_5,C_6$ correspond to the pairwise odds ratios that each depend on one $L_i$. Finally the 3-way interaction term $\Gamma(R_1,R_2,R_3|L)$ must be independent of $L$ since there can be no potential function depending on $L$ as just mentioned; this term corresponds to $C_7$. It follows that the numerator, and thus the normalizing constant, is identified. We now give a more general algebraic proof of the Theorem for arbitrary $K$ without relying on any graphical concepts. 

\begin{proof}[Proof of Theorem~\ref{thm:missingnessID}]
	We need to show that the numerator (and thus the normalizing constant in the denominator) is a function of only observed data. It is immediate from (\ref{eq:NSC}) that for each $i$, $p(R_i|R_{-i}=1,L) = p(R_i|R_{-i}=1,L_{-i})$ is a function of observed data. For each $i \in \{3,...,K\}$, 
	\begin{equation*}
	\begin{aligned}
	&\Odds(R_i,(R_1,...,R_{i-1})|(R_{i+1},...,R_K) = 1,L)\\
	&= \Odds(R_i,R_{i-1}|R_{-(i-1,i)} = 1,L) \Odds(R_i,(R_1,...,R_{i-2})|(R_{i+1},...,R_K) = 1,R_{i-1},L)
	\end{aligned}
	\end{equation*}
	where 
	$R_{-(i-1,i)}$ is shorthand for $(R_1,...,R_{i-2},R_{i+1},...,R_K)$. We may apply this expansion inductively (i.e., expand out the second term in the product) for all odds ratios that appear in $\prod_{i=2}^K \Odds(R_i,(R_1,...,R_{i-1})|(R_{i+1},...,R_K) = 1,L)$ 
	and re-arrange terms to eventually arrive at a representation of the numerator in terms of pairwise odds ratios and higher-order interaction terms, i.e., the numerator takes the form:
	\begin{equation*}
	\begin{aligned}
		&\prod_{i=1}^{K} p(R_i | R_{-i}=1,L) \prod_{(i,j)} \Odds(R_i,R_j|R_{-(i,j)}=1,L) \times \\
		 &\prod_{(i,j,k)} \Gamma(R_i,R_j,R_k|R_{-(i,j,k)}=1,L) \prod_{(i,j,k,l)} \Gamma(R_i,R_j,R_k,R_l|R_{-(i,j,k,l)}=1,L) \times \dots \times \Gamma(R_1,...,R_K | L)
	\end{aligned}
	\end{equation*}
	where the $\Gamma(\cdot | \cdot, L)$ are 3-way, 4-way, ..., and $K$-way interaction terms. These interaction terms take the following forms: $\Gamma(R_i,R_j,R_k|R_{-(i,j,k)}=1,L) = \frac{\Odds(R_i,R_j|R_k, R_{-(i,j,k)}=1,L)}{\Odds(R_i,R_j|R_k=1, R_{-(i,j,k)}=1,L)}$, $\Gamma(R_i,R_j,R_k,R_l|R_{-(i,j,k,l)}=1,L) = \frac{\Odds(R_i,R_j|R_k, R_l, R_{-(i,j,k,l)}=1,L)}{\Odds(R_i,R_j|R_k=1, R_l, R_{-(i,j,k,l)}=1,L)}\frac{\Odds(R_i,R_j|R_k=1, R_l=1, R_{-(i,j,k,l)}=1,L)}{\Odds(R_i,R_j|R_k, R_l=1, R_{-(i,j,k,l)}=1,L)}$, and so on, up to $\Gamma(R_1,...,R_K | L) =\\ \frac{\Odds(R_i,R_j|R_{-(i,j)},L) \prod_{(k,l)}\Odds(R_i,R_j|R_{(k,l)}=1,R_{-(i,j,k,l)},L) \prod_{(k,l,m,n)}\Odds(R_i,R_j|R_{(k,l,m,n)}=1,R_{-(i,j,k,l,m,n)},L) \times \cdots }{\prod_{k} \Odds(R_i,R_j|R_k=1,R_{-(i,j,k)},L) \prod_{(k,l,m)}\Odds(R_i,R_j|R_{(k,l,m)}=1,R_{-(i,j,k,l,m)},L) \times \cdots }$.
	 
	 Each pairwise odds ratio is a function of only $L_{-(i,j)}$ and hence a function of observed data by the following symmetry argument. Consider the pairwise odds ratio
	\begin{equation*}\begin{aligned}
	\Odds(R_i,R_j|R_{-(i,j)}=1,L) &= \frac{p(R_i|R_{-(i,j)}=1,R_j,L)p(R_i=1|R_{-i}=1,L)}{p(R_i|R_{-i}=1,L)p(R_i=1|R_{-(i,j)}=1,R_j,L)} \\
	&= \frac{p(R_i|R_{-(i,j)}=1,R_j,L_{-i})p(R_i=1|R_{-i}=1,L_{-i})}{p(R_i|R_{-i}=1,L_{-i})p(R_i=1|R_{-(i,j)}=1,R_j,L_{-i})}
	\end{aligned}
	\end{equation*}
	by (\ref{eq:NSC}). By symmetry, 
	\begin{equation*}\begin{aligned}
	\Odds(R_i,R_j|R_{-(i,j)}=1,L) &= \frac{p(R_j|R_{-(i,j)}=1,R_i,L)p(R_j=1|R_{-j}=1,L)}{p(R_j|R_{-j}=1,L)p(R_j=1|R_{-(i,j)}=1,R_i,L)} \\
	&= \frac{p(R_j|R_{-(i,j)}=1,R_i,L_{-j})p(R_j=1|R_{-j}=1,L_{-j})}{p(R_j|R_{-j}=1,L_{-j})p(R_j=1|R_{-(i,j)}=1,R_i,L_{-j})}
	\end{aligned}
	\end{equation*}
	so $\Odds(R_i,R_j|R_{-(i,j)}=1,L) = \Odds(R_i,R_j|R_{-(i,j)}=1,L_{-(i,j)})$ a function of only $L_{-(i,j)}$. A similar symmetry argument establishes that each $\Gamma(R_i,...,R_t | R_{-(i,...,t)} = 1,L)$ is 
	a function of $L_{-(i,...,t)}$ 
	and thus also a function of observed data. Each $\Gamma(\cdot | \cdot, L)$ function is symmetric in its arguments to the left of the conditioning bar, and so has multiple equivalent representations.
	For example, consider a 3-way term, $\Gamma(R_1,...,R_3 | (R_{4},...,R_K) = 1,L) = \frac{\Odds(R_1,R_3|(R_{4},...,R_K) = 1, R_2,L)}{\Odds(R_1,R_3|(R_2, R_{4},...,R_K) = 1, L)}$ $= \frac{\Odds(R_1,R_2|(R_{4},...,R_K) = 1, R_3,L)}{\Odds(R_1,R_2|(R_3,...,R_K) = 1, L)} = \frac{\Odds(R_2,R_3|(R_{4},...,R_K) = 1, R_1,L)}{\Odds(R_2,R_3|(R_1, R_{4},...,R_K) = 1, L)}$. By (\ref{eq:NSC}), $\frac{\Odds(R_1,R_3|(R_{4},...,R_K) = 1, R_2,L)}{\Odds(R_1,R_3|(R_2, R_{4},...,R_K) = 1, L)}$ is a function of $L_{-(1,3)}$, $\frac{\Odds(R_1,R_2|(R_{4},...,R_K) = 1, R_3,L)}{\Odds(R_1,R_2|(R_3,...,R_K) = 1, L)}$ is a function of $L_{-(1,2)}$, and $\frac{\Odds(R_2,R_3|(R_{4},...,R_K) = 1, R_1,L)}{\Odds(R_2,R_3|(R_1, R_{4},...,R_K) = 1, L)}$ is a function of $L_{-(2,3)}$ which is only satisfied if $\Gamma(R_1,...,R_3 | (R_{4},...,R_K) = 1,L)$
	is a function of only $L_{-(1,2,3)}$, hence a function of observed data. The same is true for higher-order interaction functions which are all symmetric. The $K$-way interaction $\Gamma(R_1,...,R_K | L)$ may be equivalently expressed with any choice of indices $i,j$ in the odds ratios that comprise it and so by the same argument is not a function of $L$ at all.
\end{proof}

Next we describe the semiparametric theory necessary for Lemma \ref{lma:oddsknown} and Theorem \ref{thm:nonparametricIF} in detail. 

Let $\Lambda_1$ denote the observed nuisance tangent space for $\eta$ and $\Lambda_2$ the observed nuisance tangent space for $\gamma$. (All nuisance tangent spaces we discuss here are subspaces of the standard Hilbert space of random functions with mean zero and finite variance, equipped with the usual inner product.) We use the superscript $\perp$ to denote the orthocomplement of a space and the subscript $0$ to denote the subspace that is mean zero. In particular $\Lambda^{\perp}_{1,0} = (\Lambda_1^{\perp})_0$. We are interested in $\Lambda^{\perp}_0 = \Lambda^{\perp}_{1,0} \cap \Lambda^{\perp}_{2,0}$ since the influence functions for $\beta$ reside in this space, or more precisely the set of influence functions of all regular and asymptotically linear (RAL) estimators of $\beta$ is $\{ \E[\phi S^{'}_{\beta}]^{-1}\phi: \phi \in \Lambda^{\perp}_0\}$, where $S_{\beta}$ is the observed data score for $\beta$ evaluated at the truth. We use the following result from \citet[Prop.\ A1.3]{rotnitzky1997analysis}:
\begin{equation*}
\Lambda^{\perp}_{1,0} = \left \{ \frac{\mathbb{I}(R=1)}{\pi_J(L)} d(L) + a(R, L_{(R)}) : d(L) \in \Lambda^{F,\perp}_0, a(R, L_{(R)}) \in \Lambda^A \right \}
\end{equation*}
where $\Lambda^{F,\perp}_0$ is the orthocomplement to the nuisance tangent space in the full model, and $\Lambda^A = \{ a(R, L_{(R)}) : \E[a(R, L_{(R)})|L]=0 \}$. In our setting, since the influence functions for $\beta$ in the full model are denoted by $\phi_{\text{full}}(\beta)$ we have that $d(L) = \phi_{\text{full}}(\beta)$. 
We can also characterize $\Lambda^A$ by noting that it can be written equivalently as\\ $\Lambda^A = \left \{ \sum_{j = 1}^{J-1} \mathbb{I}(R=r_j) a(R, L_{(R)}) - \frac{\mathbb{I}(R=1)}{\pi_J(L)} \sum_{j = 1}^{J-1} \pi_j(L) a(R, L_{(R)}) \right \}$. To see that they are equivalent, first one may verify that an element of this set is mean zero conditional on $L$. Then, we only need to argue the converse, that every conditionally mean zero function $a(R, L_{(R)})$ is in this set. $\E[a(R, L_{(R)}) | L] = 0$ implies that $\pi_J(L) a(1,L_{(R)}) = -\sum_{j=1}^{J-1} \pi_j(L) a(R, L_{(R)})$. Since $a(R, L_{(R)}) = \sum_{j=1}^{J-1} \mathbb{I}(R=r_j) a(R, L_{(R)}) + \mathbb{I}(R=1) a(1, L_{(R)})$ by substitution the result follows. Therefore,
\begin{equation*}
\Lambda^{\perp}_{1,0} = \left \{ \frac{\mathbb{I}(R=1)}{\pi_J(L)}  \phi_{\text{full}}(\beta) + \sum_{j=1}^{J-1} \mathbb{I}(R=r_j) a(R, L_{(R)}) - \frac{\mathbb{I}(R=1)}{\pi_J(L)} \sum_{j=1}^{J-1} \pi_j(L) a(R, L_{(R)}) \right \}
\end{equation*}
for any function $a(R, L_{(R)})$ of the observed data. \citet[p.~1117]{scharfstein1999adjusting} establish:
\begin{equation*}
\Lambda^{\perp}_{2,0} = \left \{ d(R,L_{(R)}) : d(R,L_{(R)}) \in \Lambda^{F, \perp}_{2,0} \right \}
\end{equation*}
where $d$ is some function and $\Lambda^{F}_{2}$  would be the nuisance tangent space for $\gamma$ in the full data model. As noted in the main text, we use a multinomial logistic parameterization for the conditional model, $\log \frac{\pi_j(L)}{\pi_J(L)} = h_j(L; \gamma)$. The likelihood is then:
\begin{equation*}
\begin{aligned}
\mathcal{L}_{R|L} &= \prod_{j=1}^{J} \pi_j(L)^{\mathbb{I}(R=r_j)}\\
\end{aligned}
\end{equation*}
where $\pi_j(L) = \frac{e^{h_j(L)}}{1 + \sum_{j=1}^{J-1} e^{h_j(L)}}$ for $j = 1,...,J-1$ and $\pi_J(L) = \frac{1}{1 + \sum_{j=1}^{J-1} e^{h_j(L)}}$.
The set of scores for the conditional model can be obtained by taking the logarithm and derivative w.r.t.\ $\gamma$. Consequently the nuisance tangent space is:
\begin{equation*}
\Lambda^{F}_{2,0} = \left \{ \sum_{j=1}^{J-1} \Big( \mathbb{I}(R=r_j) - \pi_j(L) \Big) g_j(L) \right \}
\end{equation*}
for all functions $g_j$. 
For Lemma \ref{lma:oddsknown} we assume that the log odds ratio is a known function of $L$, i.e., that $h_j(L; \gamma) = h_{1,j}(L=0;\gamma) + h_{2,j}(L)$ where $h_{2,j}$ is known. Then, taking derivatives w.r.t.\ $\gamma$ gives a simpler nuisance tangent space for the full model:
\begin{equation*}
\Lambda^{F}_{2,0} = \left \{ \sum_{j=1}^{J-1} \Big( \mathbb{I}(R=r_j) - \pi_j(L) \Big) c_j \right \}
\end{equation*}
for all constants $c_j$. Following \citet[p.~S14]{sun2016semiparametric}, we derive the intersection $\Lambda^{\perp}_0$ by noting 
\begin{equation*}
\Lambda^{\perp}_0 = \left \{ \lambda^{\perp}_{1,0} \in  \Lambda^{\perp}_{1,0} : \E[\lambda^F_{2,0} \lambda^{\perp}_{1,0}] =0, \lambda^F_{2,0} \in \Lambda^F_{2,0} \right \}.
\end{equation*} 

\begin{proof}[Proof of Lemma~\ref{lma:oddsknown}]
	We show that the nuisance tangent space $\Lambda^{\perp}_0$ consists of the single element
	\begin{equation*}
	\Lambda^{\perp}_{0} = \left \{ \frac{\mathbb{I}(R=1)}{\pi_J(L)}  \phi_{\text{full}}(\beta) + \sum_{j=1}^{J-1} \mathbb{I}(R=r_j) a^* - \frac{\mathbb{I}(R=1)}{\pi_J(L)} \sum_{j = 1}^{J-1} \pi_j(L) a^* : a^* = \E[ \phi_{\text{full}}(\beta) |R, L_{(R)}] \right \}
	\end{equation*}
	It suffices to show that this choice of $a^*$ is the unique solution to $\E[\lambda^F_{2,0} \lambda^{\perp *}_{1,0}] = 0$ for all $\lambda^{F}_{2,0} \in \Lambda^{F}_{2,0}$.
	
	\begin{equation*}\begin{split}
	\E \left \{ \left \{ \sum_{j=1}^{J-1} \Big( \mathbb{I}(R=r_j) - \pi_j(L) \Big) c_j \right \} \right. \\
	\left. \times \left \{ \frac{\mathbb{I}(R=1)}{\pi_J(L)}  \phi_{\text{full}}(\beta) + \sum_{j = 1}^{J-1} \mathbb{I}(R=r_j) a^* - \frac{\mathbb{I}(R=1)}{\pi_J(L)} \sum_{j = 1}^{J-1} \pi_j(L) a^* \right \} \right \} = 0 \hspace{2mm} \forall c_j \\
	\iff \E \left \{ \left \{ \sum_{j=1}^{J-1} \Big( \mathbb{I}(R=r_j) - \pi_j(L) \Big) \right \} \right. \\
	\left. \times \left \{ \frac{\mathbb{I}(R=1)}{\pi_J(L)}  \phi_{\text{full}}(\beta) + \sum_{j = 1}^{J-1} \mathbb{I}(R=r_j) a^* - \frac{\mathbb{I}(R=1)}{\pi_J(L)} \sum_{j = 1}^{J-1} \pi_j(L) a^* \right \} \right \} = 0 \\
	\iff \E \left \{ \frac{\mathbb{I}(R=1) \sum_{j=1}^{J-1} \Big( \mathbb{I}(R=r_j) - \pi_j(L) \Big) }{\pi_J(L)}  \phi_{\text{full}}(\beta) \right \}  \\ 
	+ \E \left \{ \sum_{j=1}^{J-1} \Big( \mathbb{I}(R=r_j) - \pi_j(L) \Big) \sum_{j = 1}^{J-1} \mathbb{I}(R=r_j) a^*  \right \}  \\
	- \E \left \{ \frac{\mathbb{I}(R=1)}{\pi_J(L)} \sum_{j=1}^{J-1} \Big( \mathbb{I}(R=r_j) - \pi_j(L) \Big)  \sum_{j = 1}^{J-1} \pi_j(L) a^* \right \} = 0 \\
	\iff -\E \left \{ \sum_{j = 1}^{J-1} \pi_j(L)  \phi_{\text{full}}(\beta) \right \}
	+ \E \left \{ (1 - \sum_{j = 1}^{J-1} \pi_j(L)) \sum_{j = 1}^{J-1} \pi_j(L) a^* \right \} \\
	+ \E \left \{ \sum_{j = 1}^{J-1} \pi_j(L)  \sum_{j = 1}^{J-1} \pi_j(L) a^* \right \} = 0 \\
	\iff -\E \left \{ \sum_{j = 1}^{J-1} \pi_j(L)  \phi_{\text{full}}(\beta) \right \} + \E \left \{ \sum_{j = 1}^{J-1} \pi_j(L) a^* \right \} = 0
	\end{split}
	\end{equation*}
	\begin{align*} \begin{split}
	&\iff \E \left \{ \sum_{j = 1}^{J-1} \pi_j(L) \Big( a^* - \E[  \phi_{\text{full}}(\beta) | R, L_{(R)} ] \Big)  \right \} = 0 \\
	&\iff a^* = \E[ \phi_{\text{full}}(\beta) |R, L_{(R)}]
	\end{split}
	\end{align*}
	The third $\iff$ follows because $\E \left \{ \frac{\mathbb{I}(R=1) \sum_{j=1}^{J-1} ( \mathbb{I}(R=r_j) - \pi_j(L) ) }{\pi_J(L)} |L \right \} = \E \left \{ \frac{-\mathbb{I}(R=1) \sum_{j=1}^{J-1} \pi_j(L)}{\pi_J(L)} |L \right \}$ $= - \sum_{j=1}^{J-1} \pi_j(L)$ and $\E \left \{ \sum_{j=1}^{J-1} \Big( \mathbb{I}(R=r_j) - \pi_j(L) \Big) \sum_{j = 1}^{J-1} \mathbb{I}(R=r_j) | L  \right \}$\\ $= \sum_{j=1}^{J-1} \pi_j(L) - \sum_{j=1}^{J-1} \pi_j(L) \sum_{j=1}^{J-1} \pi_j(L) = (1 - \sum_{j=1}^{J-1} \pi_j(L)) \sum_{j=1}^{J-1} \pi_j(L)$.
\end{proof}

This gives the set of influence functions when $h_{2,j}$ is known. Next, we loosen this restriction and derive the nonparametric IF when the $\Odds$ is not a known function of $L$. We need to ``adjust'' the above estimating function by a term which accounts for nonparametric estimation of the $\Odds$. We follow the approach of \citet{sun2016semiparametric} who consider a parametric model for the missingness mechanism and an instrumental variable assumption. (Note, however, that this strategy is a general one, which can be applied to any MNAR problem so long as the odds ratio is identified.) 

Consider the estimating function $\phi_{\text{odds}}(\beta; \Odds) \in \Lambda_0^{\perp}$ we derived, which is a function of $\beta$ and the (known) the odds ratio. $\E[\phi_{\text{odds}}(\beta; \Odds)]=0$ at the truth. We can consider parametric submodels by indexing the densities that appear in the odds ratio with $t$ such that $p_t(\cdot)|_{t=0} = p(\cdot)$, i.e., the true densities are recovered at $t=0$. $\E_t[\phi_{\text{odds}}(\beta_t; \Odds_t)]=0$ for all $t$. Then:
\begin{equation}\begin{aligned}
\nabla_t \Et[\phi_{\text{odds}}(\beta_t; \Oddst)]=0\\
\E[\phi_{\text{odds}}(\beta; \Odds) S(R,L_{(R)})] + \E[\nabla_t \phi_{\text{odds}}(\beta_t; \Oddst)] =0\\
\E[\phi_{\text{odds}}(\beta; \Odds) S(R,L_{(R)})] + \E \left[ \frac{\partial}{\partial \beta_t} \phi_{\text{odds}}(\beta_t; \Odds) \frac{\partial \beta_t}{\partial t} \right] +  
\E[ \nabla_t \phi_{\text{odds}}(\beta; \Oddst)]=0
\end{aligned}
\end{equation}
Solving for $\frac{\partial \beta_t}{\partial t}$:
\begin{equation}
\frac{\partial \beta_t}{\partial t} = - \E \left[ \frac{\partial}{\partial \beta_t} \phi_{\text{odds}}(\beta_t; \Odds) \right]^{-1} \times \left( \E[\phi_{\text{odds}}(\beta; \Odds) S(R,L_{(R)})] + \E[ \nabla_t \phi_{\text{odds}}(\beta; \Oddst)] \right)
\end{equation}
If we can write the second term in this sum as the expectation of some function times a score, i.e., $\E[\phi_{\text{adj}}(\beta; \Odds)S(R,L_{(R)})]$, then our desired IF is $ - \E \left[ \frac{\partial}{\partial \beta} \phi_{\text{odds}}(\beta; \Odds) \right]^{-1} \times \left( \phi_{\text{odds}}(\beta; \Odds) + \phi_{\text{adj}}(\beta; \Odds) \right) $. That is, we ``adjust'' the previously found estimating function from Lemma \ref{lma:oddsknown} and normalize by $ - \E \left[ \frac{\partial}{\partial \beta} \phi_{\text{odds}}(\beta; \Odds) \right]^{-1}$.

\begin{proof}[Proof of Theorem~\ref{thm:nonparametricIF}]
	\begin{align*}
	\E[ \nabla_t \phi(\beta; \Oddst)] &= \E \left[ \nabla_t \frac{\mathbb{I}(R=1)}{\pi_J(L;t)}  \phi_{\text{full}}(\beta) \right] \\
	&+ \E \left[ \nabla_t \sum_{j=1}^{J-1} \mathbb{I}(R=r_j) \Et[ \phi_{\text{full}}(\beta) |R, L_{(R)}] \right] \\
	&- \E \left[ \nabla_t \frac{\mathbb{I}(R=1)}{\pi_J(L;t)} \sum_{j = 1}^{J-1} \pi_j(L;t) \Et[ \phi_{\text{full}}(\beta) |R, L_{(R)}] \right]\\
	&= \E \left[ \nabla_t \frac{\mathbb{I}(R=1)}{\pi_J(L;t)}  \phi_{\text{full}}(\beta) \right] \\
	&+ \E \left[ \sum_{j=1}^{J-1} \mathbb{I}(R=r_j) \nabla_t \Et[ \phi_{\text{full}}(\beta) |R, L_{(R)}] \right] \\
	&- \E \left[ \mathbb{I}(R=1) \sum_{j = 1}^{J-1} \nabla_t \frac{\pi_j(L;t)}{\pi_J(L;t)} \E[ \phi_{\text{full}}(\beta) |R, L_{(R)}] \right]\\
	&- \E \left[ \mathbb{I}(R=1) \sum_{j = 1}^{J-1}  \frac{\pi_j(L;t)}{\pi_J(L;t)} \nabla_t \Et[ \phi_{\text{full}}(\beta) |R, L_{(R)}] \right]
	\end{align*}
	\begin{align*}
	&= \E \left[ \nabla_t \frac{\mathbb{I}(R=1)}{\pi_J(L;t)}  \phi_{\text{full}}(\beta) \right] \\
	&- \E \left[ \mathbb{I}(R=1) \sum_{j = 1}^{J-1} \nabla_t \frac{\pi_j(L;t)}{\pi_J(L;t)} \E[ \phi_{\text{full}}(\beta) |R, L_{(R)}] \right]\\
	&= \E \left[ \nabla_t \frac{\mathbb{I}(R=1) \sum_{j = 1}^{J-1} \pi_j(L;t)}{\pi_J(L;t)}  \phi_{\text{full}}(\beta) \right] \\
	&- \E \left[ \mathbb{I}(R=1) \sum_{j = 1}^{J-1} \nabla_t \frac{\pi_j(L;t)}{\pi_J(L;t)} \E[ \phi_{\text{full}}(\beta) |R, L_{(R)}] \right]\\
	&= \E \left[ \mathbb{I}(R=1) \sum_{j = 1}^{J-1} \nabla_t \frac{\pi_j(L;t)}{\pi_J(L;t)} \Delta(R,L) \right]
	\end{align*} 
	where $\Delta(R,L) \equiv  \phi_{\text{full}}(\beta) - \E[ \phi_{\text{full}}(\beta) |R, L_{(R)}]$. Then, using the definition of $\Odds(r,L)$:
	\begin{align*}
	\E[ \nabla_t \phi(\beta; \Oddst)] &= \E \left[ \mathbb{I}(R=1) \sum_{j = 1}^{J-1} \nabla_t \Oddst(r,L) \frac{\pi_j(L=0)}{\pi_J(L=0)} \Delta(R,L) \right]\\
	&= \E \left[ \mathbb{I}(R=1) \sum_{j = 1}^{J-1} \frac{ \nabla_t \Oddst(r,L)}{\Odds(r,L)} \frac{\pi_j(L)}{\pi_J(L)} \Delta(R,L) \right]\\
	&= \E \left[ \sum_{j = 1}^{J-1} \nabla_t \log \Oddst(r,L) \pi_j(L)\Delta(R,L) \right]\\
	&= \E \left[ \sum_{j = 1}^{J-1} \sum_{i=1}^K (1-r_i) \nabla_t \delta h_i(L_{-i};t) \pi_j(L)\Delta(R,L) \right]\\
	&= \E \left[ \sum_{i=1}^K (1-R_i) \nabla_t \delta h_i(L_{-i};t) \Delta(R,L) \right]
	\end{align*}
	\begin{align*}
	&= \E \left[ \sum_{i=1}^K \E[(1-R_i)|L_{-i}] \nabla_t \delta h_i(L_{-i};t) \E[\Delta(R,L)|R_i=0,L_{-i}] \right]\\
	&= -\E \left[ \sum_{i=1}^K \E[(1-R_i)|L_{-i}] \frac{ p(R_i=1|R_{-i},L_{-i})}{ p(R_i=0|R_{-i},L_{-i})} \frac{\nabla_t p_t(R_i=1|R_{-i}=1,L)}{p(R_i=1|R_{-i}=1,L)^2} \right. \\
	& \qquad \qquad \left. \times \E[\Delta(R,L)|R_i=0,L_{-i}] \Big] \right. (*) \\
	&= -\E \left[ \sum_{i=1}^K \E[(1-R_i)|L_{-i}] \frac{\mathbb{I}(R_{-i}=1)}{p(R_{-i}=1|L_{-i})} \left( \frac{R_i}{p(R_i=1|R_{-i},L_{-i})}-1 \right) \right. \\
	& \qquad \qquad \left. \times \vphantom{\sum_{i=1}^K} \frac{ p(R_i=1|R_{-i},L_{-i})}{ p(R_i=0|R_{-i},L_{-i})} \frac{\nabla_t p_t(R_i|R_{-i},L)}{p(R_i|R_{-i},L)} \E[\Delta(R,L)|R_i=0,L_{-i}] \right] \\
	&= -\E \left[ \sum_{i=1}^K \E[(1-R_i)|L_{-i}] \frac{\mathbb{I}(R_{-i}=1)}{p(R_{-i}=1|L_{-i})} \left( \frac{R_i}{p(R_i=1|R_{-i},L_{-i})}-1 \right)  \right. \\
	& \qquad \qquad \left. \times \vphantom{\sum_{i=1}^K} \frac{ p(R_i=1|R_{-i},L_{-i})}{ p(R_i=0|R_{-i},L_{-i})} \E[\Delta(R,L)|R_i=0,L_{-i}] S(R_i|R_{-i},L) \right] \\
	&= -\E \left[ \sum_{i=1}^K \E[(1-R_i)|L_{-i}] \frac{\mathbb{I}(R_{-i}=1)}{p(R_{-i}=1|L_{-i})} \left( \frac{R_i}{p(R_i=1|R_{-i},L_{-i})}-1 \right) \right. \\
	& \qquad \qquad \left. \times \vphantom{\sum_{i=1}^K} \frac{ p(R_i=1|R_{-i},L_{-i})}{ p(R_i=0|R_{-i},L_{-i})} \E[\Delta(R,L)|R_i=0,L_{-i}] S(R,L) \right] \\
	&= -\E \left[ \sum_{i=1}^K \E[(1-R_i)|L_{-i}] \frac{\mathbb{I}(R_{-i}=1)}{p(R_{-i}=1|L_{-i})} \left( \frac{R_i}{p(R_i=1|R_{-i},L_{-i})}-1 \right) \right. \\
	& \qquad \qquad \left. \times \vphantom{\sum_{i=1}^K}  \frac{  p(R_i=1|R_{-i},L_{-i})}{ p(R_i=0|R_{-i},L_{-i})} \E[\Delta(R,L)|R_i=0,L_{-i}] \E[S(R,L)|R,L_{(R)}] \right] \\
	\end{align*}
	where at (*) we used $\nabla_t \left( \log(\frac{p_t(R_i=0|R_{-i}=1,L_{-i})}{p_t(R_i=1|R_{-i}=1,L_{-i})}) - \log(\frac{p(R_i=0|R_{-i}=1,L_{-i}=0)}{p(R_i=1|R_{-i}=1,L_{-i}=0)}) \right) =\\
	\frac{\nabla_t \frac{p_t(R_i=0|R_{-i}=1,L_{-i})}{p_t(R_i=1|R_{-i}=1,L_{-i})} }{\frac{p_t(R_i=0|R_{-i}=1,L_{-i})}{p_t(R_i=1|R_{-i}=1,L_{-i})}} = \frac{p_t(R_i=1|R_{-i}=1,L_{-i})}{p_t(R_i=0|R_{-i}=1,L_{-i})} \left( \frac{- \nabla_t p_t(R_i=1|R_{-i}=1,L_{-i})}{p(R_i=1|R_{-i}=1,L_{-i})} - \frac{ p(R_i=0|R_{-i}=1,L_{-i}) \nabla_t p_t(R_i=1|R_{-i}=1,L_{-i})}{p(R_i=1|R_{-i}=1,L_{-i})^2} \right)\\ =
	- \frac{p_t(R_i=1|R_{-i}=1,L_{-i})}{p_t(R_i=0|R_{-i}=1,L_{-i})} \nabla_t p_t(R_i=1|R_{-i}=1,L_{-i}) \left( \frac{p(R_i=1|R_{-i}=1,L_{-i})}{p(R_i=1|R_{-i}=1,L_{-i})^2} + \frac{p(R_i=0|R_{-i}=1,L_{-i})}{p(R_i=1|R_{-i}=1,L_{-i})^2} \right)\\ =
	- \frac{p_t(R_i=1|R_{-i}=1,L_{-i})}{p_t(R_i=0|R_{-i}=1,L_{-i})}  \left( \frac{\nabla_t p_t(R_i=1|R_{-i}=1,L_{-i})}{p(R_i=1|R_{-i}=1,L_{-i})^2} \right).$
\end{proof}

\begin{proof}[Proof of Theorem~\ref{thm:DR}]
	First, consider the model $\mathcal{M}^{\pi} = \cap_j \mathcal{M}_j^{\pi}$ where $\mathcal{M}_j^{\pi} = \mathcal{M}_{\pi,j} \cap \mathcal{M}_{\Odds}$. Then:
	\begin{align*}
	&\E \left \{ \frac{\mathbb{I}(R=1)}{\pi_J(L,X)}  \phi_{\text{full}}(\beta) + \sum_{j=1}^{J-1} \mathbb{I}(R=r_j) \E[ \phi_{\text{full}}(\beta)|R, L_{(R)}, X] \right.\\
	& \hspace{5mm} \left. - \frac{\mathbb{I}(R=1)}{\pi_J(L,X)} \sum_{j = 1}^{J-1} \pi_j(L,X) \E[\phi_{\text{full}}(\beta)|R, L_{(R)}, X] \right \} = \\
	&\E \left \{ \phi_{\text{full}}(\beta) \right \} + \E \left \{ \sum_{j=1}^{J-1} \mathbb{I}(R=r_j) \E[\phi_{\text{full}}(\beta)|R, L_{(R)}, X] \right \}\\
	& \hspace{5mm} - \E \left \{  \sum_{j = 1}^{J-1} \pi_j(L,X) \E[\phi_{\text{full}}(\beta)|R, L_{(R)}, X] \right \} = \\
	&\E \left \{ \phi_{\text{full}}(\beta) \right \} = 0.
	\end{align*}
	Next, consider the model $\mathcal{M}^{\text{PM}} = \cap_j \mathcal{M}_j^{\text{PM}}$ where $\mathcal{M}_j^{\text{PM}} = \mathcal{M}_{\text{PM},j} \cap \mathcal{M}_{\Odds}$. Notice that 
	\begin{align*}
	\frac{1}{\pi_J(L,X)} = 1 + \sum_{j=1}^{J-1} \frac{\pi_j(L,X)}{\pi_J(L,X)} = 1 + \sum_{j=1}^{J-1} \Odds(r_j,L)\frac{\pi_j(L=0,X)}{\pi_J(L=0,X)}.
	\end{align*}
	Then:
	\begin{align*}
	&\E \left \{ \frac{\mathbb{I}(R=1)}{\pi_J(L,X)} \phi_{\text{full}}(\beta) + \sum_{j=1}^{J-1} \mathbb{I}(R=r_j) \E[\phi_{\text{full}}(\beta)|R, L_{(R)}, X] \right. \\
	&\left. \hspace{5mm} - \frac{\mathbb{I}(R=1)}{\pi_J(L,X)} \sum_{j = 1}^{J-1} \pi_j(L,X) \E[\phi_{\text{full}}(\beta)|R, L_{(R)}, X] \right \} = \\
	&\E \left \{ \mathbb{I}(R=1) \left( 1 + \sum_{j=1}^{J-1} \Odds(r_j,L)\frac{\pi_j(L=0,X)}{\pi_J(L=0,X)} \right) \phi_{\text{full}}(\beta) \right \} \\
	& \hspace{5mm} + \E \left \{ \sum_{j=1}^{J-1} \mathbb{I}(R=r_j) \E[\phi_{\text{full}}(\beta)|R, L_{(R)}, X] \right \} \\
	& \hspace{5mm} - \E \left \{ \frac{\mathbb{I}(R=1)}{\pi_J(L,X)} \sum_{j = 1}^{J-1} \pi_j(L,X) \E[\phi_{\text{full}}(\beta)|R, L_{(R)}, X] \right \} = \\
	&\E \left \{ \mathbb{I}(R=1) \phi_{\text{full}}(\beta) \right \} 
	+ \E \left \{ \sum_{j=1}^{J-1} \mathbb{I}(R=r_j) \E[\phi_{\text{full}}(\beta)|R, L_{(R)}, X] \right \} \\
	& \hspace{5mm} + \E \left \{ \mathbb{I}(R=1) \sum_{j = 1}^{J-1} \Odds(r_j,L)\frac{\pi_j(L=0,X)}{\pi_J(L=0,X)} \left( \phi_{\text{full}}(\beta) - \E[\phi_{\text{full}}(\beta)|R, L_{(R)}, X] \right) \right \} = \\
	&\E \left \{ \mathbb{I}(R=1) \phi_{\text{full}}(\beta) \right \} 
	+ \E \left \{ \sum_{j=1}^{J-1} \mathbb{I}(R=r_j) \E[\phi_{\text{full}}(\beta)|R, L_{(R)}, X] \right \} \\
	& \hspace{5mm} + \E \left \{ \mathbb{I}(R=1) \sum_{j = 1}^{J-1} \Odds(r_j,L)\frac{\pi_j(L=0,X)}{\pi_J(L=0,X)} \left( \phi_{\text{full}}(\beta) - \frac{\E[\Odds(r_j,L)\phi_{\text{full}}(\beta)|R=1, L_{(r_j)}, X]}{\E[\Odds(r_j,L)|R=1, L_{(r_j)}, X]} \right) \right \} = \\
	&\E \left \{ \mathbb{I}(R=1) \phi_{\text{full}}(\beta) \right \} 
	+ \E \left \{ \sum_{j=1}^{J-1} \mathbb{I}(R=r_j) \E[\phi_{\text{full}}(\beta)|R, L_{(R)}, X] \right \} = \\
	&\E \left \{ \sum_{j=1}^{J} \mathbb{I}(R=r_j) \E[\phi_{\text{full}}(\beta)|R, L_{(R)}, X] \right \} =
    \E \left\{ \phi_{\text{full}}(\beta) \right\} = 0
	\end{align*}
	The result for the union model follows.
\end{proof}

\begin{proof}[Proof of Theorem~\ref{thm:estimator}]
	By the previous theorem we have a system of consistent estimating equations for $\Omega = (\beta, \psi_X, \mu, \Odds)'$ and so according to standard semiparametric theory we have (under the usual regularity conditions) $\sqrt{n}(\widehat{\Omega} - \Omega_0) = -\sqrt{n} \E[\frac{\partial V(\Omega_0)}{\partial \Omega_0}]^{-1} \mathbb{P}_n (V(\Omega_0)) + o_p(1)$ \citep{van2000asymptotic}. The result follows.
\end{proof}

\newpage

\bibliographystyle{plainnat}
\bibliography{../references-personal/bib}


\end{document}